\documentclass[11pt,a4paper]{article}

\usepackage[margin=1in]{geometry}
\usepackage{graphicx}
\usepackage{amsthm}
\usepackage{amsmath}
\usepackage{amssymb}
\usepackage{verbatim}
\usepackage{algorithmic}
\usepackage{algorithm}
\usepackage{color}
\usepackage{enumerate}
\usepackage{appendix}

\DeclareMathOperator{\Tr}{Tr}

\newcommand{\longer}[1]{#1 }
\renewcommand{\longer}[1]{ }

\newcommand{\eps}{\varepsilon}

\newcommand{\LAS}{\text{\sc{Las}}}

\newcommand{\PS}{\mathcal{P}}

\newcommand{\M}{\mathcal{M}}
\newcommand{\K}{\mathcal{K}}
\newcommand{\D}{{\delta}}
\newcommand{\PD}{\mathcal{PD}}
\newcommand{\ND}{\mathcal{ND}}

\newtheorem{theorem}{Theorem}[section]

\newtheorem{lemma}[theorem]{Lemma}
\newtheorem{corollary}[theorem]{Corollary}

\newtheorem{definition}{Definition}[section]
\newtheorem{example}{Example}[section]

\newtheorem{remark}{Remark}[section]

\title{The Lasserre Hierarchy in Almost Diagonal Form 
\thanks{Supported by the Swiss National Science Foundation project 200020-144491/1 ``Approximation Algorithms for Machine Scheduling Through Theory and Experiments''.}}

\author{ Monaldo Mastrolilli \\
IDSIA, monaldo@idsia.ch \\
 Lugano, Switzerland}
\begin{document}
\maketitle
\begin{abstract}
The Lasserre hierarchy is a systematic procedure for constructing a sequence of increasingly tight relaxations that capture the convex formulations used in the best available approximation algorithms for a wide variety of optimization problems. Despite the increasing interest, there are very few techniques for analyzing Lasserre integrality gaps. Satisfying the positive semi-definite requirement is one of the major hurdles to constructing Lasserre gap examples.

We present a novel characterization of the Lasserre hierarchy based on moment matrices that differ from diagonal ones by matrices of rank one (\emph{almost diagonal form}).
We provide a modular \emph{recipe} to obtain positive semi-definite feasibility conditions by iteratively diagonalizing rank one matrices.

Using this, we prove strong lower bounds on integrality gaps of Lasserre hierarchy for two basic capacitated covering problems. 
For the min-knapsack problem, we show that the integrality gap remains arbitrarily large even at level $n-1$ of Lasserre hierarchy. For the min-sum of tardy jobs scheduling problem, we show that the integrality gap is unbounded at level $\Omega(\sqrt{n})$ (even when the objective function is integrated as a constraint). These bounds are interesting on their own, since both problems admit FPTAS. 
\end{abstract}


\section{Introduction}

The use of mathematical programming relaxations, such as linear programming (LP) and semidefinite programming (SDP), has been one of the most powerful tools in approximation algorithms.
These algorithms are analyzed by comparing the value of the returned integer solution to that of the fractional solution.
The \emph{integrality gap} is the quotient of the true optimum by the relaxationÕs
optimal solution (which is always at least as good), and measures the quality of the approximation.

The integrality gap is sensitive to the original integer programming formulation, and an important question is when modifications to the integer program improve the algorithms of this framework.
This has lead to systematic procedures, also known as \emph{lift-and-project} methods, for constructing a sequence of increasingly tight mathematical programming relaxation. In particular,
Lov\'{a}sz-Schrijver+ (LS+)~\cite{Lov91} and the stronger Lasserre~\cite{Lasserre01} semidefinite programming hierarchy, and
Lov\'{a}sz-Schrijver (LS)~\cite{Lov91}, and the stronger Sherali-Adams (SA) hierarchy~\cite{SheraliA90} for linear programs were created to systematically improve semidefinite and linear programs at the cost of additional runtime (see \cite{Laurent03} for a comparison).  

Introduced in 1999 by Grigoriev and Vorobjov~\cite{GrigorievV01}, the ÒSum of SquaresÓ (SOS) proof system is a powerful algebraic proof system. Shor~\cite{schor87}, Lasserre~\cite{Lasserre01} and Parrilo~\cite{parrilo00} show that this proof system is automatizable using semidefinite programming (SDP), meaning that any $n$-variable degree-$d$ proof can be found in time $n^{O(d)}$. Furthermore, the SDP is dual to the Lasserre SDP hierarchy, meaning that the ``d/2-round Lasserre value'' of an optimization
problem is equal to the best bound provable using a degree-$d$ SOS proof (see also the monograph by Laurent~\cite{laurent09}). For brevity, we will interchange Lasserre hierarchy with SOS hierarchy since they are essentially the same in our context. For a brief history of the different formulations from \cite{GrigorievV01}, \cite{Lasserre01}, \cite{parrilo00} and the relations between them and results in real algebraic geometry we refer the reader to \cite{ODonnellZ13}.

These hierarchies provably imply some of the most celebrated approximation algorithms for NP-complete problems even after a few rounds. For example, the first round of LS+ (and hence also Lasserre) for the Independent Set problem implies the Lov\'{a}sz $\theta$-function~\cite{Lovasz79} and for the MaxCut problem gives the Goemans-Williamson relaxation~\cite{GoemansW95}.
The ARV relaxation of the SparsestCut \cite{AroraRV09} problem is no stronger than the relaxation given in the third round of LS+ (and hence also Lasserre), and most recently the subexponential time algorithm for
Unique Games~\cite{AroraBS10} is implied by a sublinear number of rounds of Lasserre~\cite{BarakRS11,GuruswamiS11}. Other improved approximation guarantees that arise from the first $O(1)$ levels of Lasserre (or weaker) hierarchy can be found in~\cite{BarakRS11,BateniCG09,Chlamtac07,ChlamtacS08,VegaK07,GuruswamiS11,MagenM09,RaghavendraT12}.
For a more detailed overview on the use of hierarchies in approximation algorithms, see the recent survey of Chlamt\'{a}\v{c} and Tulsiani~\cite{Chla12}, the reading group web-page in~\cite{bansalweb} and the references therein.

Integrality gap results for Lasserre are thus very strong unconditional negative results, as they apply to a ``model of computation" that includes the best known approximation algorithms for several problems (see~\cite{Chla12} for a discussion). If a large integrality gap persists after a large number of rounds then a wide class of efficient approximation algorithms are ruled out (this implicitly contains some of the most sophisticated approximation algorithms for many problems). 

SOS hierarchy appears to be more powerful than those of LS+ and SA+ hierarchies: a recent work of Barak et al. \cite{BarakBHKSZ12} proved that $O(1)$ number of rounds of the SOS hierarchy can solve the unique games problem on instances which need a non-constant number of LS+, SA+ rounds. These results emphasize that a better understanding of the power and limitations of the SOS hierarchy is necessary.

Most of the known lower bounds for the SOS hierarchy originated in the works of Grigoriev~\cite{Grigoriev01b,Grigoriev01} (also independently rediscovered by Schoenebeck~\cite{Schoenebeck08}). These works show that random 3XOR or 3SAT instances cannot be solved by even $\Omega(n)$ rounds of SOS hierarchy. Subsequent lower bounds, such as those of~\cite{Tulsiani09}, \cite{BhaskaraCVGZ12} rely on \cite{Grigoriev01b}, \cite{Schoenebeck08} plus gadget reductions.

%
%

An interesting line of research is given by answering the following questions~\cite{KarlinMN11}: ``How strong are these restricted models of computation with respect to approximation? In other words, how much do lower bounds in these models tell us about the intrinsic hardness of the problems studied?''.
One set of weaknesses revolves around the fact that SOS has a hard time reasoning about $X_1 + ... + X_n$ using the fact that all $X_i$'s are integers. In~\cite{Grigoriev01}, it is assumed that $X_1, ..., X_n$ are +/-1 and that $n$ is odd. Then degree-$(n-1)$ SOS cannot disprove $X_1 + ... + X_n = 0$, even though of course $|X_1 + ... + X_n| \geq 1$. (A simplified proof can be found in~\cite{GrigorievHP02}). 

In \cite{Cheung07}, Cheung considered the case of $X_1, ..., X_n$ constrained to be 0/1 (i.e., $X_i^2 = X_i$) along with the constraint $X_1 + ... + X_n \geq \delta$. Cheung showed that there exists $\delta = \delta(n) > 0$ such that with this constraint, degree-$(n-1)$ Lasserre/SOS cannot prove $X_1 + ... + X_n \geq 1$, even though that is of course true. Cheung was motivated by the ``Knapsack'' polytope. He shows that degree-$(n-1)$ SOS cannot refute $X_1 + ... + X_n = 1 - 1/(n+1)$. (In terms of SOS, the result presented in this paper for min-knapsack implies that degree-$(n-1)$ SOS cannot refute a much weaker formula, i.e. $X_1 + ... + X_n = 1/k$, where $k$ can be arbitrarily small.)
Additional results and references can be found in the monograph by Laurent~\cite{laurent09}.

Karlin et al. \cite{KarlinMN11}, focused on one of the most basic \emph{packing problem} that is well-known to be ``easy'' from the viewpoint of approximability, namely the \emph{maximum knapsack problem}\footnote{Given a set of items, each with a weight and a value, determine the number of each item to include in a collection so that the total weight is less than or equal to a given limit and the total value is as large as possible.} which is well-known \cite{IbarraK75,Lawler79} to admit a fully polynomial time approximation scheme (FPTAS).
While from the perspective of approximation algorithms, there is nothing to be gained by applying convex optimization techniques to this problem, it is a useful tool for gaining a better understanding of the strengths and properties of various hierarchies of relaxations.
They show two results.
First they prove that an integrality gap close to 2 persists up to a linear number of rounds of Sherali-Adams. This confirms that Sherali-Adams restricted model of computation has serious weakness and a lower bound in this model does not necessarily imply that it is difficult to get a good approximation ratio (this has been observed also in other contexts, see e.g. \cite{CharikarMM09}).
On the other side, they show that after $r^2$ rounds of Lasserre, the integrality gap decreases quickly from $2$ to $r/(r-1)$, implying as a side product of their analysis a polynomial time approximation scheme. 
To some extent, their second result gives some further evidence that Lasserre's hierarchy might be seen as an effective model of computation for certain \emph{packing} problem. The results presented in this paper show that this does not seem the case for basic capacitated \emph{covering} problems.

\paragraph{Capacitated Covering Problems Relaxations.}
The \emph{min-knapsack problem} is defined by a set of items, each with a cost and a value, and a specified demand. The goal is to select a minimum cost set of items with total value at least the demand. The minimum knapsack problem is well-known to admit an FPTAS (the FPTAS for maximum knapsack \cite{IbarraK75,Lawler79}  can be easily modified to work for the min-knapsack problem). This fundamental covering problem is a special case of many covering problems, including the very general capacitated covering integer program \cite{CarrFLP00}. 

The \emph{capacitated covering integer program} (see e.g.~\cite{CarrFLP00}) is an integer program of the form $\min\{cx:Ux\geq d, 0\leq x\leq b, x\in Z^+\}$, where all the entries of $c,U$ and $d$ are nonnegative. 
The min-knapsack problem, the capacitated network design problem and the min-sum of tardy jobs scheduling problem are only few examples of capacitated covering problems.
One difficulty in approximating the capacitated covering problems lies in the fact that the ratio between the optimal IP solution to the optimal LP solution can be as bad as $||d||_{\infty}$, even when $U$ consists of a single row (i.e. the min-knapsack problem). 

A powerful way to cope with this problem is to strengthen the LP by adding (exponentially many) knapsack cover (KC) inequalities introduced by Carr et al.~\cite{CarrFLP00}, that have proved to be a useful tool to address capacitated covering problems~\cite{BansalBN08,CarnesS08,LeviLS08,BansalGK10,ChakrabartyGK10}.
For the min-knapsack problem, the improved IP/LP ratio with these inequalities is $2$.


%
%

The min-sum single-machine scheduling problem (often denoted $1||\sum f_j$) is defined by a set of $n$ jobs to be scheduled on a single machine. Each job has an integral processing time, and there is a monotone function $f_j(C_j)$ specifying the cost incurred when job $j$ is completed at a particular time $C_j$; the goal is to minimize $\sum_j f_j(C_j)$.
A natural special case of this problem is given by the min-sum of tardy jobs (denoted $1||\sum_j w_jT_j$), where $f_j(C_j)= w_j \max\{C_j-d_j,0\}$, $w_j\geq 0$, and $d_j>0$ is a specified due date of job $j$. This problem is known to be NP-complete \cite{yuan92} even for unit weights. 
FPTASÕs are known with the additional restriction that there are only a constant number of deadlines \cite{KarakostasKW12}, or if jobs have unit weights \cite{lawler82}. 

The first constant approximation algorithm for $1||\sum f_j$ (and for $1||\sum_j w_jT_j$) has been obtained by Bansal and Pruhs \cite{BansalP10} (they consider an even more general scheduling problem). Their $16$-approximation has been recently improved to a $(2+\eps)$ primal-dual approximation by Cheung and Shmoys \cite{CheungS11}. Both approaches are based on using capacitated covering linear program relaxations (the relaxations in \cite{BansalP10,CheungS11} are different), with unbounded integrality gap (because the min-knapsack LP is a special case in both cases). Thus, in \cite{BansalP10,CheungS11} the authors strengthen these LPs by adding the knapsack cover (KC) inequalities introduced in~\cite{CarrFLP00}. Based on this approach, the combinatorial primal-dual approach in \cite{CheungS11} is currently the best known result. 

However,
no hardness of approximation result is known for $1||\sum f_j$ and, as remarked in \cite{CheungS11}, ``it is still conceivable (and perhaps likely) that there exists a polynomial time approximation scheme''. With this aim, better lower bounds are sought since (KC) inequalities are not sufficient: Indeed, even for the very special case of the min-knapsack problem, the integrality gap of the LP augmented with (KC) inequalities is $2$~\cite{CarrFLP00}. 

On the other side, for the minimum knapsack problem, using the trick of ``lifting the objective function'' (i.e. when the objective function is integrated as a constraint), Karlin et al. results imply that the Lasserre SDP hierarchy reduces the integrality gap to $(1+\eps)$ at level $O(1/\eps)$, for any $\eps>0$.

In light of the latter result, it is therefore natural to understand what happens if we strengthen the capacitated covering LPs with some levels of the Lasserre Hierarchy, instead of using (KC) inequalities. (Note that in order to claim that one can optimize over Lasserre hierarchy in polynomial time, one needs to assume that  the number of constraint of the starting LP is polynomial in the number of variables (see the discussion in~\cite{Laurent03}). Therefore LP cannot be the linear program strengthened with the exponentially many (KC) inequalities).




\subsection{Our Results}
The contribution of this paper is twofold: 
\begin{enumerate}
\item (Almost Diagonal Matrices) We provide a novel characterization of the Lasserre hierarchy based on almost diagonal matrices.
Using this, we present an iterative recipe for computing positive semi-definite feasibility conditions.
\item (Integrality Gaps) We provide strong Lasserre lower bounds for basic capacitated covering problems that admit FPTASs.
\end{enumerate}

\paragraph{Almost Diagonal Matrices.}

It is known~\cite{Lasserre01} that the $0/1$ polytope is found at level $n$ of the Lasserre hierarchy. This argument concerns an elementary property of the zeta matrix of the lattice given by the collection of all subsets of $N=\{1,\ldots,n\}$ (see ~\cite{Laurent03}). It is enlightening to revisit this analysis by emphasizing some important aspects (Section~\ref{Sect:leveln}). We will use similar arguments and translate them to a generic level $t$ ($\leq n$) in a very natural way to obtain almost diagonal moment matrices (Section~\ref{Sect:almostdiag}), i.e. matrices that differ from diagonal ones by rank one matrices.

One of the main challenge in analyzing gap examples for Lasserre hierarchy is given by positive-semidefinite constraints, and by the hurdles of checking whether a solution satisfies them.
Indeed, it is well-known that there is no explicit general formula for computing the eigenvalues of a matrix $A$ (this because for polynomials of degree $5$ and higher there is no formula for computing the roots in terms of the coefficient in a finite number of steps; The eigenvalues of $A$ are, of course, the roots of the characteristic polynomial of $A$). 
Nonetheless, there are very effective \emph{iterative algorithms} for computing the eigenvalues of a symmetric matrix. The original iterative algorithm for this purpose was devised by Jacobi (see e.g. \cite{Strang}).
Jacobi's idea is to use the similarity transform that diagonalizes a $2\times 2$ matrix (for which a closed formula exists) to \emph{partially diagonalize} any $n\times n$ matrix with the aim to reduce the ``norm'' of the off-diagonal entries.

By using almost diagonal moment matrices, we suggest an iterative \emph{recipe} for computing positive semi-definite feasibility conditions (Section~\ref{Sect:gershgorin}). These conditions are used for showing that certain solutions are feasible for the Lasserre Hierarchy. 
Our iterative approach can be seen as a reminiscent of Jacobi's algorithm. One of the main difference is that we diagonalize the matrices of rank one\footnote{Diagonalizing a rank one matrix boils down to reducing it to a zero matrix but one diagonal entry.} that appears in the almost diagonal form (instead of diagonalizing $2\times 2$ matrices). Again the goal is to reduce the ``importance'' of the off-diagonal entries or, in other words, the radii of Gershgorin disks~\cite{varga}. The Lasserre integrality gap constructions of Sections~\ref{sect:minknapsack} and~\ref{sect:knaplifted} are based on this technique. 

Finally, starting from the almost diagonal form, we suggest an alternative formulation of the Lasserre hierarchy as a semi-infinite linear program (Section~\ref{sect:sip}). This gives a non-matricial definition of the hierarchy and a different point of view that might be convenient for certain problems. The Lasserre integrality gap construction of Section~\ref{sect:tardyjobs} is based on this formulation.

We think that the Lasserre hierarchy in almost diagonal form has several interesting aspects that will be useful in other applications and will stimulate further research. The proposed approach belongs to the very few techniques known so far for proving Lasserre integrality gaps.



\paragraph{Integrality Gaps.}
By using Lasserre in almost diagonal form, we prove strong Lasserre lower bounds for basic capacitated covering problems that admit FPTASs.


When we do not ``lift the objective function'', we show (Section~\ref{sect:minknapsack}) that the integrality gap for the min-knapsack remains arbitrarily large even at level $(n-1)$ of Lasserre's hierarchy (note that this is a tight characterization, since at level $n$ the solution is integral). 

If we ``lift the objective function'', we show (Section~\ref{sect:tardyjobs}) that the integrality gap of the Lasserre hierarchy for the min-sum of tardy jobs scheduling problem is unbounded at level $t=\Omega(\sqrt{n})$. The standard covering LP that we use here is a common special case of the covering LPs used in \cite{BansalP10,CheungS11}, and therefore it shows that the approach used in \cite{BansalP10,CheungS11} cannot be improved by simply replacing the (KC) inequalities with the Lasserre Hierarchy at level $O(\sqrt{n})$. The same gap analysis holds for the min version of the multiple knapsacks problem (see Section \ref{sect:knaplifted}), and for the capacitated network design problem defined in~\cite{CarrFLP00}, by a straightforward modification of the gap construction and the same analysis.) Trivially the gap bounds immediately apply to the capacitated covering integer program and to all the subproblems for which the min-knapsack is just a special case (see e.g.~\cite{CarrFLP00}). 

We note that most of prior results exhibiting gap instances for Lasserre hierarchy relaxations do so for
problems that are already known to be hard to approximate, under some suitable assumption. Based
on this hardness result, one would expect that the Lasserre hierarchy relaxations to have an integrality gap that
matches the inapproximability factor. Some exceptions are also known where the known integrality gaps are substantially stronger than the (very weak) hardness bounds known for the problem (see \cite{BhaskaraCVGZ12} and the references therein), but here it is still conceivable that the apparent ``weakness'' of the Lasserre hierarchy is due to the inherent complexity of the problem, that has still to be fully understood, and are perhaps indicative of the hardness of approximating.

In this paper, our gap constructions are a rare exception to this trend, indeed we show unbounded integrality gaps for two ``easy'' problems that admit FPTASs. These results give some evidence that Lasserre restricted model of computation has serious and extreme weakness for covering problems of this type and might stimulate the study of better hierarchies (see e.g.~\cite{BienstockZ04}).

By quoting   \cite{barakweb}:
``While until recently we had very little tools to take advantage of the SOS algorithm
(at least in the sense of having rigorous analysis), we now have some indications
that, when applied to the \emph{right problems} it can be a powerful toolÉ''. 
We believe that the provided integrality gaps help to shape the above sentence.

\paragraph{How to read this paper.} 
Most of the concepts and technical aspects in this paper are anticipated by concrete examples (see examples~\ref{ex:def}, \ref{examplediag}, \ref{example:almostdiag}, \ref{ex:strategy}, \ref{exampleknap}, \ref{ex:mkp} and Section~\ref{sect:exmkp}) and high level expositions, with the aim to provide the reader with the essential intuition. Generalizations of the examples and formal proofs are then subsequently presented.
The non-expert reader can get the main sense of the content by reading only the definitions and the provided examples (and skip the remaining part). The expert-reader might skip some of the examples.


\section{The Lasserre Hierarchy: Definition}

In this section, we provide a definition of the Lasserre hierarchy~\cite{Lasserre01}. 
With a slight cost in notation, the system is introduced in its generality (and not tailored to the studied problems). The reason of this choice is because the almost diagonal form derived in this paper (see Lemma~\ref{th:lassalmostdiag}) holds for the general case and we believe it will be useful for other problems as well.

In our notation, we mainly follow the
survey of Laurent~\cite{Laurent03}. We also provide some well-known properties with the aim to be self-contained and use well-known facts from linear algebra (see e.g. \cite{Strang} and Section~\ref{linear algebra}). 

\paragraph{Variables and Moment Matrix.} Throughout this paper, vectors are written as columns. 
Let $N$ denote the set $\{1,\ldots,n\}$. The collection of all subsets of $N$ is denoted by $\PS(N)$. For any integer $t\geq 0$, let $\PS_t(N)$ denote the collection of subsets of $N$ having cardinality at most~$t$.
Let $y\in \mathbb{R}^{\PS(N)}$. For convenience, $y_{\{i\}}$ is abbreviated as $y_i$ for all $i\in V$. Let $Diag(y)$ denote the diagonal matrix in $\mathbb{R}^{\PS(N)\times \PS(N)}$ with $(I,I)$-entry equal to $y_I$ for all $I\in \PS(N)$. For any nonnegative integer $t\leq n$, let $M_t(y)$ denote the matrix with $(I,J)$-entry $y_{I\cup J}$ for all $I,J\in \PS_t(N)$. 
Matrix $M_n(y)$ is known as \emph{moment matrix} of $y$.
%

%

\paragraph{Lasserre Hierarchy Definition.}
Let $\K$ be defined by the following
\begin{equation}\label{eq:polytopeK}
\K:=\{x\in[0,1]^n:g_{\ell}(x)\geq 0 \text{ for } \ell=1,\ldots,m\}
\end{equation}
where $g_{\ell}$ is a non constant polynomial\footnote{In our applications $g_{\ell}$ is a linear function in $x$.} in $x$ for $\ell=1,\ldots,m$.
We are interested in obtaining the convex hull of the integral points in $\K$, therefore we can assume that each variable occurs in every polynomial $g_{\ell}$ with degree at most one, since $x_i^2=x_i$ for every $i\in N$ when $x$ is integral.

Given a polynomial $g(x)$, we use the same symbol $g$ to denote the vector in $\mathbb{R}^{\PS(N)}$ where the entry indexed by $I$ is equal to the coefficient of the term $\prod_{i\in I} x_i$ in $g(x)$, for all $I\in \PS(N)$ and, therefore, $g(x) = \sum_{I\subseteq N}\left(g_I \prod_{i\in I} x_i\right )$.

For $g,y\in \mathbb{R}^{\PS(N)}$, we define $g*y := M_N(y) g$, often called \emph{shift operator}\footnote{For any polynomial $g(x)$ the shift operation $(g*y)_I$ is obtained by ``linearizing'' the polynomial $g(x)\prod_{i\in I}x_i$.}; that is, the $I$-th entry of vector $g*y$, namely $(g*y)_I$, is equal to $\sum_{J\subseteq N} g_J y_{I\cup J}$.\footnote{
When $g(x)$ is a linear function of $x$, i.e. $g(x) = \sum_{i=1}^n g_{i} \cdot x_i + g_{\emptyset}$ we have $(g*y)_I=\sum_{i=1}^n g_i y_{I\cup \{i\}}+g_I$.}

We will use $A\succeq 0$ to denote that matrix $A$ is positive semi-definite (PSD) (see e.g. \cite{Strang} and Appendix~\ref{linear algebra}).
\begin{definition}\label{lassDef}
The Lasserre hierarchy at the $t$-th level, denoted as $\LAS_t(\K)$, is the set of vectors $y\in \mathbb{R}^{\PS(N)}$ that satisfy the following
\begin{eqnarray}
y_{\emptyset}&=&1  \\
M_{t+1}(y)&\succeq& 0  \\
M_{t}( g_{\ell}*y )&\succeq& 0 \qquad \ell=1\ldots m
\end{eqnarray}
\end{definition}
In the following we will call informally $M_{t+1}(y)$ the \emph{Variable Moment Matrix}, and $M_{t}( g_{\ell}*y )$ the \emph{Constraint Moment Matrix} (the context will make it clear the level $t$ we are referring to).

\begin{example}\label{ex:def}
Here we introduce a small example that will be used to gain the essential intuition of several properties before their formal exposition. The expert reader can skip these parts.

The example consists of a linear program with one constraint and two variables: 
\begin{equation}\label{eq:expolytope}
\K:=\{x_1,x_2\in[0,1]:g(x)=a_1x_1+a_2x_2-b\geq 0\}
\end{equation}
At the $2$-nd level  of Lasserre we obtain the \emph{full} variable and constraint moment matrices.
The variable moment matrix is as follows.
{\small 
\begin{equation}\label{ex:My}
M_2(y) =\left(
\begin{array}{cccc}
y_{\emptyset} & y_{1} & y_{2} & y_{12} \\
y_{1} & y_{1} & y_{12} & y_{12} \\
y_{2} & y_{12} & y_{2} & y_{12} \\
y_{12} & y_{12} & y_{12} & y_{12}
\end{array}
\right) \succeq 0
\end{equation}
}
Similarly, let $z=g*y$, the constraint moment matrix is as follows:
{\small 
\begin{equation}\label{ex:Mz}
M_2(z) =\left(
\begin{array}{cccc}
z_{\emptyset} & z_{1} & z_{2} & z_{12} \\
z_{1} & z_{1} & z_{12} & z_{12} \\
z_{2} & z_{12} & z_{2} & z_{12} \\
z_{12} & z_{12} & z_{12} & z_{12}
\end{array}
\right) \succeq 0
\end{equation}
}
where $z_{\emptyset}=a_1y_1+a_2y_2-b$, $z_{1}=a_1y_1+a_2y_{12}-by_1$, $z_{2}=a_1y_{12}+a_2y_{2}-by_2$ and $z_{12}=(a_1+a_2-b)y_{12}$.
\end{example}
\section{The Lasserre Hierarchy at Level $n$}\label{Sect:leveln}

With $n$ variables, the $n$-th level of the Lasserre hierarchy is sufficient to obtain a tight relaxation where the only feasible solutions are convex combinations of integral solutions~\cite{Lasserre01}.
This can be proved by using the \emph{canonical lifting lemma} (see Laurent~\cite{Laurent03}), which characterizes when a full moment matrix $M_n(y)$ is positive semidefinite by diagonalizing it. In the following we revisit this result (see~\cite{Laurent03} for additional details), since it will be useful for introducing the almost diagonal form. 

Informally, the properties of the moment matrices at level $n$ are ``revealed'' by diagonalizing; We will use ``partial'' diagonalization to ``reveal'' the properties of moment matrices at level $t$.
%
%
\begin{example}\label{examplediag}
By using Example~\ref{ex:def}, we present the essential core properties  that are used in showing convergence to the integral hull of the Lasserre hierarchy. 

It is well-known that elementary symmetric (i.e. row and column) matrix operations preserve the positive-semidefiniteness of a matrix. These transformations are known under the term \emph{congruent} transformations (see Appendix \ref{linear algebra} and the notation therein). For example, consider matrix \eqref{ex:Mz}. Remove from the first row, the second and the third row, and add the last row (and, symmetrically, do the same for the first column), we obtain the following congruent matrix that is PSD if and only if $M_2(z)$ is PSD (see Lemma~\ref{th:simop}).
$$
{\small
M_2(z) \cong \left(
\begin{array}{cccc}
z_{\emptyset}-z_1-z_2+z_{12} & 0 & 0 & 0 \\
0 & z_{1} & z_{12} & z_{12} \\
0 & z_{12} & z_{2} & z_{12} \\
0 & z_{12} & z_{12} & z_{12}
\end{array}
\right) \succeq 0
}$$
Now, remove the last row from the second row (symmetrically for columns); then  remove the last row from the third row  (simmetrically for columns). Then, we obtain the following diagonal matrix. 
{\small $$
M_2(z) \cong \left(
\begin{array}{cccc}
z_{\emptyset}-z_1-z_2+z_{12} & 0 & 0 & 0 \\
0 & z_{1}-z_{12} & 0 & 0 \\
0 & 0 & z_{2}-z_{12} & 0 \\
0 & 0 & 0 & z_{12}
\end{array}
\right) \succeq 0
$$}
By playing a bit, it is not very difficult to understand the general rule to diagonalize any full moment matrix. This transformation can be suitably described by multiplying the moment matrix by a special matrix that is known as the M\"{o}bius matrix of $\mathcal{P}(N)$ (see~\cite{Laurent03}, below and the following sections).
{\tiny
$$\underbrace{\left(
\begin{array}{cccc}
1 & -1 & -1 & 1 \\
0 & 1 & 0 & -1 \\
0 & 0 & 1 & -1 \\
0 & 0 & 0 & 1
\end{array}
\right)}_{Z^{-1}=\mbox{M\"{o}bius matrix}} 
\left(
\begin{array}{cccc}
z_{\emptyset} & z_{1} & z_{2} & z_{12} \\ 
z_{1} & z_{1} & z_{12} & z_{12} \\
z_{2} & z_{12} & z_{2} & z_{12} \\ 
z_{12} & z_{12} & z_{12} & z_{12}
\end{array}
\right)
\underbrace{
\left(
\begin{array}{cccc}
1 & 0 & 0 & 0 \\
-1 & 1 & 0 & 0 \\
-1 & 0 & 1 & 0 \\
1 & -1 & -1 & 1
\end{array}
\right)}_{(Z^{-1})^{\top}}
 = \left(
\begin{array}{cccc}
z_{\emptyset}-z_1-z_2+z_{12} & 0 & 0 & 0 \\
0 & z_{1}-z_{12} & 0 & 0 \\
0 & 0 & z_{2}-z_{12} & 0 \\
0 & 0 & 0 & z_{12}
\end{array}
\right)
$$
}
{\tiny
$$
\left(
\begin{array}{cccc}
z_{\emptyset} & z_{1} & z_{2} & z_{12} \\ 
z_{1} & z_{1} & z_{12} & z_{12} \\
z_{2} & z_{12} & z_{2} & z_{12} \\ 
z_{12} & z_{12} & z_{12} & z_{12}
\end{array}
\right)
 = \underbrace{\left(
\begin{array}{cccc}
1 & 1 & 1 & 1 \\
0 & 1 & 0 & 1 \\
0 & 0 & 1 & 1 \\
0 & 0 & 0 & 1
\end{array}
\right)}_{Z\mbox{ matrix}} \left(
\begin{array}{cccc}
z_{\emptyset}-z_1-z_2+z_{12} & 0 & 0 & 0 \\
0 & z_{1}-z_{12} & 0 & 0 \\
0 & 0 & z_{2}-z_{12} & 0 \\
0 & 0 & 0 & z_{12}
\end{array}
\right)\underbrace{
\left(
\begin{array}{cccc}
1 & 0 & 0 & 0 \\
1 & 1 & 0 & 0 \\
1 & 0 & 1 & 0 \\
1 & 1 & 1 & 1
\end{array}
\right)}_{Z^{\top}}
$$
}

Similarly we can diagonalize the full moment matrix \eqref{ex:My} of the variables:
{\small $$
M_2(y) \cong \left(
\begin{array}{cccc}
y_{\emptyset,\{1,2\}} & 0 & 0 & 0 \\
0 & y_{\{1\},\{2\}} & 0 & 0 \\
0 & 0 & y_{\{2\},\{1\}} & 0 \\
0 & 0 & 0 & y_{\{1,2\},\{\emptyset\}}
\end{array}
\right) \succeq 0
$$}
where $y_{\emptyset,\{1,2\}}=y_{\emptyset}-y_1-y_2+y_{12}$, $y_{\{1\},\{2\}}=y_{1}-y_{12}$, $y_{\{2\},\{1\}}=y_{2}-y_{12}$ and $y_{\{1,2\},\{\emptyset\}}=y_{12}$.  

Note that $M_2(y)\succeq 0$ if and only if the above diagonal matrix is PSD, and therefore when all the diagonal entries are nonnegative. Moreover, the sum of the diagonal entries is equal to $y_{\emptyset}=1$. So, the diagonal entries form a probability distribution. For example, $y_{\emptyset,\{1,2\}}$ denotes the probability that neither variable $x_1$ nor $x_2$ are set to one; $y_{\{1\},\{2\}}$ denotes the probability that we set to one only variable $x_1$ (and zero $x_2$). So we have a probability distribution over all the integral solutions.

It is also very instructive to give a closer look at the diagonalization of the constraint moment matrix. Consider the first diagonal entry: $z_{\emptyset}-z_1-z_2+z_{12}= (a_1y_1+a_2y_2-b)-(a_1y_1+a_2y_{12}-by_1)-(a_1y_{12}+a_2y_{2}-by_2)+(a_1+a_2-b)y_{12}=(-b)(y_{\emptyset}-y_1-y_2+y_{12})=-b \cdot y_{\emptyset,\{1,2\}}$. So the first entry is equal to the value of the constraint when none of the variables is set to one multiplied by the corresponding probability. It is not difficult to verify that similar things happen to the other diagonal entries (see Lemma~\ref{th:constentry}) and we obtain the following: 
{\small 
$$
M_2(z) \cong \left(
\begin{array}{cccc}
(-b)y_{\emptyset,\{1,2\}} & 0 & 0 & 0 \\
0 & (a_1-b)y_{\{1\},\{2\}} & 0 & 0 \\
0 & 0 & (a_2-b)y_{\{2\},\{1\}} & 0 \\
0 & 0 & 0 & (a_1+a_2-b)y_{\{1,2\},\{\emptyset\}}
\end{array}
\right) \succeq 0
$$
}
It follows that $M_2(z)$ is PSD if and only if the diagonal matrix is PSD, which implies that we have positive probability for an integral solution if and only if the constraint is satisfied by that integral solution. Note that $y_1= y_{\{1\},\{2\}}+y_{\{1,2\},\emptyset}$ (similarly for $y_2$) can be seen as convex combination of \emph{feasible} integral solutions and therefore we have no integrality gap.

In the following, our example with two variables is generalized to any number of variables and constraints.
\end{example}
\subsection{The Lasserre Hierarchy at Level $n$ in Diagonal Form}
Let us start by introducing some basic notations and preliminary properties. We use the generic vector $w\in \mathbb{R}^{\PS(N)}$ to denote either the vector $y\in\mathbb{R}^{\PS(N)}$ of variables, or the shifted vector $g*y$, for any $g\in\mathbb{R}^{\PS(N)}$.
\begin{definition}
Let $w\in \mathbb{R}^{\PS(N)}$. For any $I,J\in \PS(N)$, we define
\begin{equation} \label{eq:probevent}
w_{I,J}: =\sum_{H\subseteq J} (-1)^{|H|} w_{H\cup I}
\end{equation}
Let $w^N\in \mathbb{R}^{\PS(N)}$ be such that the $I$-th entry, with $I\subseteq N$, is equal to
\begin{equation}\label{eq:totprob}
w^N_{I} := w_{I,N\setminus I}
\end{equation}
\end{definition}
Note that at level $n$, $y_I^N$ can be interpreted as the probability of the integral solution $\{y_i=1: i\in I, y_j=0 :j\in N\setminus I\}$. 
The following two properties are easy to check.
\begin{lemma}\label{th:nonempty}
For any $w\in \mathbb{R}^{\PS(N)}$ and $I\cap J\not= \emptyset$ we have $w_{I, J}=0$.
\end{lemma}
\longer{
\begin{proof}
For $x\in I\cap J$,
$w_{I,J} =\sum_{H\subseteq J} (-1)^{|H|} w_{H\cup I} = \sum_{H\subseteq J\setminus \{x\}} \left((-1)^{|H|}  + (-1)^{|H|+1}  \right)w_{H\cup I}$.
\end{proof}
}
%
%
\begin{lemma}\label{th:sum1}
For any $J\subseteq N$ and $w\in \mathbb{R}^{\PS(N)}$ we have
\begin{equation}
w_J = \sum_{I\subseteq N} w_{I\cup J, N\setminus I} \label{sum1}
\end{equation}
\end{lemma}
\longer{
\begin{proof}
By definition,
$\sum_{I\subseteq S} w_{I\cup J, S\setminus I} \label{sum1} = \sum_{I\subseteq S}  \sum_{H\subseteq S\setminus I} (-1)^{|H|} w_{I\cup J\cup H}$.
Now consider the above sum and the generic term of the sum, say $w_{T\cup J}$, with $T\subseteq S$. Set $T$ can be seen as the union of two disjoint sets, say $I$ and $H$, with $T=I\cup H$; variable $w_{T\cup J}$ appears for every possible pair of disjoint sets $I$ and $H$, with $T=I\cup H$, each time multiplied by coefficient $(-1)^{|H|}$. Therefore the sum of the coefficients of the generic term $w_{T\cup J}$ is equal to $\sum_{H \subseteq T}(-1)^{|H|}$; if $|T|\geq 1$ then half of the subsets $H$ have even cardinality and therefore $\sum_{H \subseteq T}(-1)^{|H|}=0$; otherwise, $T=\emptyset$ and the coefficient is equal to 1 and the claim follows:
$\sum_{I\subseteq S}  \sum_{H\subseteq S\setminus I} (-1)^{|H|} w_{I\cup J\cup H}
= \sum_{T\subseteq S} w_{T\cup J} \sum_{H\subseteq T} (-1)^{|H|} = w_J$.
\end{proof}
}

\paragraph{Diagonalization.}\label{sect:inc-esc}

Let $Z$ denote \emph{zeta matrix} of the lattice $\PS(N)$, that is the square $0$-$1$ matrix indexed by $\PS(N)$ such that $Z_{I,J}=1$ if and only if $I \subseteq J$.
\begin{equation}\label{zetamatrix}
Z_{I,J}=
\left\{
\begin{array}{ll}
1   & \text{if } I \subseteq J,\\
0 & \text{otherwise}.
\end{array} \right.
\end{equation}
%
%
This matrix is known to be invertible and the inverse is known as the M\"{o}bius matrix of $\mathcal{P}(N)$ whose entries are defined as follows:
\begin{equation}\label{mobiusmatrix}
Z^{-1}_{I,J}=
\left\{
\begin{array}{ll}
(-1)^{|J\setminus I |}   & \text{if } I \subseteq J,\\
0 & \text{otherwise}.
\end{array} \right.
\end{equation}
The diagonalization of the moment matrices is obtained by the following \emph{congruent transformation} (see Section \ref{linear algebra} and Definition~\ref{Def:congruence}): for $A\in \{M_n(y), M_n(g*y)\}$, $A\rightarrow Z^{-1}A (Z^{-1})^{\top}$, where $Z^{-1}$ is the M\"{o}bius matrix of $\mathcal{P}(N)$ (see~\cite{Laurent03}).
\begin{lemma}\label{th:diagonalization}
For any $w\in \mathbb{R}^{\PS(N)}$, $Z Diag(w^N) Z^{\top} =  M_n(w)$.
\end{lemma}
\begin{proof}
$\left(Z Diag(w^N) Z^{\top}\right)_{I,J} = \sum_{U\subseteq N} Z_{I,U} Z_{J,U}w_{U,N\setminus U}$, which is equal to $\sum_{\substack{U\subseteq N\\ I\cup J\subseteq U}} w_{U,N\setminus U}= \sum_{U\subseteq N\setminus (I\cup J)} w_{U\cup I\cup J,N\setminus U}=w_{I\cup J}$, where the latter equality follows from Lemma~\ref{th:sum1} and~\ref{th:nonempty}.
\end{proof}

By the previous lemma it follows that for $y\in \mathbb{R}^{\PS(N)}$, and for any constraint $g(x)\geq 0$, we have the following congruence transformations (recall $z = g*y$ and $z^N,y^N$ are defined by \eqref{eq:totprob}).
\begin{eqnarray}
M_n(y) &\cong& Diag(y^N)\\
M_n(z) &\cong& Diag(z^N)
\end{eqnarray}
\begin{lemma}\label{th:constentry}
For any $g,y\in \mathbb{R}^{\PS(N)}$ and $z = g*y$ we have
\begin{equation}
z^N_I = \left(\sum_{K\subseteq I}g_K\right)\cdot y^N_I
\end{equation}
\end{lemma}
\begin{proof}
By definition, $z^N_I=\sum_{H\subseteq [n]\setminus I} (-1)^{|H|} z_{H\cup I}=\sum_{H\subseteq [n]\setminus I} (-1)^{|H|} \sum_{K\subseteq N} g_K \cdot y_{K\cup H\cup I}=  \sum_{K\subseteq N} \left( g_K \sum_{H\subseteq [n]\setminus I} (-1)^{|H|} y_{K\cup H\cup I}\right)= \left( \sum_{K\subseteq I} g_K\right) \sum_{H\subseteq [n]\setminus I} (-1)^{|H|} y_{ H\cup I}$
and using the fact that for any $K\not \subseteq I$ we have $\sum_{H\subseteq [n]\setminus I} (-1)^{|H|} y_{K\cup H\cup I}=0$.
\end{proof}

\begin{lemma}
For any $g,y\in \mathbb{R}^{\PS(N)}$ and $z = g*y$ we have
\begin{eqnarray}
M_n(y)\succeq 0 &\Longleftrightarrow& \left(y^N_I \geq 0 \quad \forall I\subseteq N \right)\label{vars}\\
M_n(z)\succeq 0 &\Longleftrightarrow& \left(z^N_I = \left(\sum_{K\subseteq I}g_K\right)\cdot y^N_I \geq 0 \quad \forall I\subseteq N \right) \label{constraints}
\end{eqnarray}
\end{lemma}
It follows that if $y^N_I>0$ then $\sum_{K\subseteq I}g_K\geq 0$, i.e. the solution obtained by setting $y_K=1$ for every $K\subseteq I$ and $y_H=0$ for every $H\not \subseteq I$ satisfies constraint $g(y)\geq 0$
(viceversa if $\sum_{K\subseteq I}g_K< 0$,  then we must have $y^N_I =0$).

\longer{
\begin{example}
Let $g(x)= 3x_1+x_2-x_3 - 3$, where $g(x)\geq0$ is a constraint. By conditions~\eqref{vars} and \eqref{constraints} it follows that $y^N_I = 0$ for any $I\subseteq N$ such that $\sum_{K\subseteq I}g_K<0$. For example, for $I=\{1,3\}$ we have $\sum_{K\subseteq I}g_K = 3-1-3<0$ and therefore $y^N_{\{1,3\}}=0$.
\end{example}
}

\begin{lemma}\label{th:decompositionLeveln}
The projection on $\K$ of any feasible solution $y\in\LAS_n(\K)$, i.e. $\{y_j:j\in N\}$, can be seen as convex combination of integral solutions that are feasible for $\K$. 
\end{lemma}
\begin{proof}
For any $j\in N$,
$y_j = \sum_{I\subseteq N} y_{I\cup \{j\}, N\setminus I}
= \sum_{\substack{I\subseteq N\\ y_{I, N\setminus I}>0 }} \frac{y_{I\cup \{j\}, N\setminus I} }{y_{I, N\setminus I}} \cdot y_{I, N\setminus I}$
by Lemma~\ref{th:sum1}.
Note that  for any $y_{I, N\setminus I}>0$ we have $\frac{y_{I\cup \{j\}, N\setminus I} }{y_{I, N\setminus I}}=1$ if  $j\in I$ and zero otherwise.
\longer{
\begin{equation}
\frac{y_{I\cup \{j\}, N\setminus I} }{y_{I, N\setminus I}} =\left\{
\begin{array}{ll}
1 & \text{if } j\in I\\
0 & \text{else}.
\end{array}\right.
\end{equation}
}
By Lemma~\ref{th:sum1} we have $\sum_{I\subseteq N} y_{I, N\setminus I}=1$ and by \eqref{vars} we have $y_{I, N\setminus I}\geq 0$. It follows that solution $\{y_j:j\in N\}$ can be seen as a convex combinations of integral solutions: for every $y_{I, N\setminus I}>0 $, the integral solutions are those that are obtained by setting to one all the variables with indexes in $I$ and zero otherwise. Note that these integral solutions satisfy the constraints of $\K$ since $y_{I, N\setminus I}>0 $ implies that the constraints are satisfied by using \eqref{constraints}.
\end{proof}

\section{The Lasserre Hierarchy at Level $t$}\label{Sect:levelt}
In this section we translate the arguments of level $n$ to level $t$, for any $1\leq t \leq n$, by providing a ``partial'' diagonalization of the moment matrices. We will use these congruent transformations in the gap analyses. %
The following example introduces the main concepts.
\begin{example}\label{example:almostdiag}
Recall that $M_2(z)$ from Example~\ref{examplediag} is equal to the following.

{\tiny
$$
\left(
\begin{array}{cccc}
z_{\emptyset} & z_{1} & z_{2} & z_{12} \\ 
z_{1} & z_{1} & z_{12} & z_{12} \\
z_{2} & z_{12} & z_{2} & z_{12} \\ 
z_{12} & z_{12} & z_{12} & z_{12}
\end{array}
\right)
 = \left(
\begin{array}{cccc}
1 & 1 & 1 & 1 \\
0 & 1 & 0 & 1 \\
0 & 0 & 1 & 1 \\
0 & 0 & 0 & 1
\end{array}
\right) \left(
\begin{array}{cccc}
z_{\emptyset}-z_1-z_2+z_{12} & 0 & 0 & 0 \\
0 & z_{1}-z_{12} & 0 & 0 \\
0 & 0 & z_{2}-z_{12} & 0 \\
0 & 0 & 0 & z_{12}
\end{array}
\right)
\left(
\begin{array}{cccc}
1 & 0 & 0 & 0 \\
1 & 1 & 0 & 0 \\
1 & 0 & 1 & 0 \\
1 & 1 & 1 & 1
\end{array}
\right)
$$
}
$M_1(z)$ is obtained from $M_2(z)$ by removing the last row and column and therefore it is equal to 

{\tiny
$$
\left(
\begin{array}{ccc}
z_{\emptyset} & z_{1} & z_{2} \\ 
z_{1} & z_{1} & z_{12}  \\
z_{2} & z_{12} & z_{2}  
\end{array}
\right)
 =\underbrace{ \left(
\begin{array}{cccc}
1 & 1 & 1  \\
0 & 1 & 0  \\
0 & 0 & 1  
\end{array}
\right)}_{A} \left(
\begin{array}{ccc}
z_{\emptyset}-z_1-z_2+z_{12} & 0 & 0 \\
0 & z_{1}-z_{12} & 0 \\
0 & 0 & z_{2}-z_{12}
\end{array}
\right)\underbrace{
\left(
\begin{array}{ccc}
1 & 0 & 0  \\
1 & 1 & 0  \\
1 & 0 & 1 
\end{array}
\right)}_{A^{\top}} + z_{12}
\left(
\begin{array}{c}
1 \\
1 \\
 1
 \end{array}
\right) 
\left(
\begin{array}{ccc}
1 & 1 & 1 
\end{array}
\right)
$$
}
The inverse of $A$ is the M\"{o}bius matrix of $\mathcal{P}_1(N)$ and is equal to
 $A^{-1} {\tiny =
\left(
\begin{array}{cccc}
1 & -1 & -1  \\
0 & 1 & 0  \\
0 & 0 & 1  
\end{array}
\right)}
$
and by multiplying the left-hand side of $M_1(z)$  by $A^{-1}$ and the right-hand side by $(A^{-1})^{\top}$, we obtain the following matrix that is congruent to $M_1(z)$:
{\tiny
$$
\left(
\begin{array}{cccc}
1 & -1 & -1  \\
0 & 1 & 0  \\
0 & 0 & 1  
\end{array}
\right)
\left(
\begin{array}{ccc}
z_{\emptyset} & z_{1} & z_{2} \\ 
z_{1} & z_{1} & z_{12}  \\
z_{2} & z_{12} & z_{2}  
\end{array}
\right)
\left(
\begin{array}{cccc}
1 & 0 & 0  \\
-1 & 1 & 0  \\
-1 & 0 & 1  
\end{array}
\right)
 =\left(
\begin{array}{ccc}
z_{\emptyset}-z_1-z_2+z_{12} & 0 & 0 \\
0 & z_{1}-z_{12} & 0 \\
0 & 0 & z_{2}-z_{12}
\end{array}
\right)+ z_{12}\underbrace{
\left(
\begin{array}{cccc}
1 & -1 & -1  \\
-1 & 1 & 1  \\
-1 & 1 & 1  
\end{array}
\right)}_{\mbox{matrix of rank one}}
$$
}

Note that $M_1(z)$ is congruent to a matrix that differs from a diagonal matrix by a matrix of rank one.
In the next section we will see that this fact can be generalized to any level and for any number of variables.

The fact that the ``distance'' from the diagonal matrix can be expressed by matrices of rank one is an intriguing property and it will play a fundamental role in our analysis.  
\end{example}

\subsection{The Lasserre Hierarchy at Level $t$ in Almost Diagonal Form}\label{Sect:almostdiag}
In 1960, Wilf \cite{wilf60} introduced the concept of almost diagonal matrices: a matrix $A$ is \emph{almost diagonal} if there exists a diagonal matrix $D$ and vectors $x$ and $y$ such that $A= D + xy^{\top}$, i.e. $A$ differs from a diagonal matrix by a matrix of rank one. Generalizing this, we say that $A$ is \emph{$k$-almost diagonal} if it differs from a diagonal matrix by $k$ matrices of rank one (we will omit $k$ for brevity).

For any $w\in\mathbb{R}^{\PS(N)}$, let $Diag(w^N,t)$ be the submatrix of $Diag(w^N)$ indexed by $\PS_t(N)$.\footnote{Vector $w$ is intended to be either the vector $y\in\mathbb{R}^{\PS(N)}$ of variables or the shifted vector $g*y$ for any $g\in\mathbb{R}^{\PS(N)}$.} In the following we show that any moment matrix $M_t(w)$ is congruent to a matrix $M_t^*(w^N)$ that differs from $Diag(w^N,t)$ by $k$ matrices of rank one, where $k= |\PS(N)\setminus \PS_{t}(N)|$. This gives a different view of the Lasserre hierarchy at level~$t$. 

We will refer to this re-formulation as the \emph{Lasserre hierarchy in almost diagonal form} (and $M_t^*(w^N)$ as the \emph{almost diagonal decomposition} of $M_t(w)$). 

\begin{lemma}\label{th:lassalmostdiag}[Almost Diagonal Form]
Let $G(J)\in \mathbb{R}^{\PS_t(N)}$ be a vector with the $I$-th entry equal to 
\begin{equation*}
G(J)_{I} =\left\{
\begin{array}{ll}
(-1)^{t-|I|}{|J|-|I|-1 \choose t-|I|} & I\subseteq J\\
0 & \text{otherwise}
\end{array}
\right.
\end{equation*}
For any $w\in\mathbb{R}^{\PS(N)}$ and $t= 0,1,\ldots, n$
\begin{equation}
\boxed{
M_t(w)\cong M_t^*(w^N) = Diag(w^N,t) +\sum_{J\in \PS(N)\setminus \PS_t(N)} w_J^N R(J)} \label{eq:conalmostdiag}
\end{equation}
where $R(J) =G(J)G(J)^{\top}$ is a matrix (of rank one).
\end{lemma}
\begin{proof}
For any $t=0,1,\ldots,n$, consider the following block decomposition of the zeta matrix~\eqref{zetamatrix}:
\begin{eqnarray*}
Z&=&
\left(
\begin{array}{cc}
A{(t)} & B{(t)} \\
 C{(t)} & D{(t)}
\end{array}
\right)
\end{eqnarray*}
where $A{(t)}\in \mathbb{R}^{\PS_t(N)\times \PS_t(N)}$ is the square submatrix of $Z$ indexed by $\PS_t(N)$, and the submatrices $B(t),C(t),D(t)$ are defined accordingly.
Note that at level $t$, matrix $M_t(w)$ is equal to the square submatrix of $M_n(w)$ indexed by $\PS_t(N)$. 

Recall that $Diag(w^N,h)$ is defined as the submatrix of $Diag(w^N)$ indexed by $\PS_h(N)$. Let $Diag(w^N,\overline{h})$ be the submatrix of $Diag(w^N)$ indexed by $\PS(N)\setminus \PS_h(N)$.
It follows that
\begin{eqnarray*}
M_{t}(w) &=& A(t) Diag(w^N,{t}) A(t)^{\top}+B(t) Diag(w^N,{\overline{t}}) B(t)^{\top}
\end{eqnarray*}
Since matrix $A(t)$ is also invertible, then $M_t(w) \cong  M_t^*(w)$ where:
\begin{eqnarray}\label{eq:almostdiag}
M_t^*(w)&=&Diag(w^N,t) +G Diag(w^N,{\overline{t}}) G^{\top} \label{eq:conalmostdiag}\\
G&=&A(t)^{-1}B(t)
\end{eqnarray}
The claim follows by showing that for any $I\in \PS_t(N)$ and $J\in \PS(N)\setminus \PS_t(N)$ the $(I,J)$-th entry of matrix $G$ is equal to (in the claim $G(J)$ denotes the $J$-th column of $G$)
\begin{equation}
G_{I,J} =\left\{
\begin{array}{ll}
(-1)^{t-|I|}{|J|-|I|-1 \choose t-|I|} & I\subseteq J\\
0 & \text{otherwise}
\end{array}
\right.
\end{equation}

By definition $G(t)_{I,J} =\sum_{K\in \PS_t(N)} A(h)^{-1}_{I,K} B(h)_{K,J}$. Note that $A(h)^{-1}_{I,K}$ is different from zero (and equal to $(-1)^{|K\setminus I|}$) when $I\subseteq K$, and $B(h)_{K,J}$ is different from zero (and equal to one) when $K\subseteq J$. Then note that $\sum_{\ell=0}^w (-1)^{\ell}{n \choose \ell}=(-1)^w {n-1\choose w}$ (assuming $ {n-1 \choose w}=0$ for any $w\geq n$) and therefore
\begin{eqnarray*}
\sum_{\substack{I\subseteq K \subseteq J \\ K\in \PS_t(N)}} (-1)^{|K\setminus I|}&=&\sum_{\ell=0}^{t- |I|} (-1)^{\ell}{|J|- |I|\choose \ell}=(-1)^{t-|I|}{|J|-|I|-1 \choose t-|I|}
\end{eqnarray*}
\end{proof}
%
%
\begin{remark}
In order to avoid misinterpretation, we remark the following.
One difference between matrix $M_t(w)$ and $M_t^*(w^N)$ is that matrix $M_t(w)$ is function of variables in $\{w_I:I\in \PS_t(N)\}$, whereas matrix $M_t^*(w^N)$ is function of variables in $\{w_I^N:I\subseteq N\}$.

Every $w_I^N$ that appears in $M_t^*(w^N)$ is either $y_I^N$ (if $w=y$), or $z^N_I$ (if $z=g*y=w$) where $z^N_I= g(I)y_I^N$ and $g(I)$ denote the value of $g(x)$ when we set to 1 all the variables in $\{x_{i}:i\in I\}$ and to zero the remaining. The relationships between $w_I$ and $w_I^N$ are given by equations \eqref{eq:probevent}, \eqref{eq:totprob} and \eqref{sum1}.
So in transforming matrix $M_t(w)$ to $M_t^*(w^N)$ we have operated a change of basis. 

Note that any $y_I^N$ (also) depends on moments $y_J$ with $|J|>t$. But matrix $M_t(w)$ does not depend on moments $y_J$ with $|J|>t$, so also the congruent matrix $M_t^*(w^N)$ does not (by using \eqref{eq:probevent} one can check that higher order moments $y_J$ with $|J|>t$ cancel out).

One reason for using variables $\{y_I^N: I\subseteq N\}$ is because they have a very nice interpretation as (pseudo)probability (see \cite{BarakBHKSZ12,BarakKS13}): $y_I^N$ can be seen as the (pseudo)probability of the integral solution $\{x_{i}=1:i\in I\}$ (and zero the remaining variables); $z^N_I=g(I)y_I^N$ can be seen as the value $g(I)$ of constraint $g(x)\geq 0$ according to solution $\{x_{i}=1:i\in I\}$ multiplied by the corresponding (pseudo)probability. At level $n$, variables $\{y_I^N: I\subseteq N\}$ are actual probabilities, as already observed (see e.g. Example~\ref{examplediag}). At any level, any solution is a linear combination of these (pseudo)probabilities (see Equation \eqref{sum1}).
\end{remark}

\begin{remark}
Note that any matrix $R(J)=G(J)G(J)^{\top}$ in Lemma~\ref{th:lassalmostdiag} is PSD (see, e.g. Appendix \ref{linear algebra}). In the following we distinguish 3 parts of $M_t^*(w^N)$: the diagonal matrix $Diag(w^N,t)$ (that sometimes we call $D$ for brevity), the positive semi-definite ($\PD$) matrices (i.e. the rank one matrices $R(J)$ multiplied by \emph{positive} coefficient $w_J^N$), and the negative semi-definite ($\ND$) matrices (i.e. the rank one matrices $R(J)$ multiplied by \emph{negative} coefficient $w_J^N$).
\end{remark}

\subsubsection{Almost Diagonal Form: User Guide}\label{Sect:gershgorin}
Assume that we want to prove that a given solution $y\in\mathbb{R}^{\PS(N)}$ is a feasible solution for the Lasserre hierarchy at a certain level. For some $t\geq 0$, this boils down to checking if $M_t(w)\succeq 0$, (where, recall, vector $w$ is intended to be either the vector $y$ of variables, or the shifted vector $g*y$, for any $g\in\mathbb{R}^{\PS(N)}$).
By Lemma~\ref{th:lassalmostdiag}, this is equivalent to checking if $M_t^*(w^N)\succeq 0$, where $M_t^*(w^N)$ is the almost diagonal decomposition of $M_t(w)$.\footnote{Every $w_I^N$ that appears in $M_t^*(w^N)$ is either $y_I^N$ (if $w=y$), or $z^N_I$ (if $z=g*y=w$) where $z^N_I= g(I)y_I^N$ and $g(I)$ denote the value of $g(x)$ when we set to 1 all the variables in $\{x_{i}:i\in I\}$ and to zero the remaining.}

If $w_I^N\geq 0$ for every $I\subseteq N$, then it is straightforward to claim that $M_t^*(w^N)\succeq 0$. This simply because $M_t^*(w^N)$ is the sum of PSD matrices.

If $w_I^N\geq -\eps$ for every $I\subseteq N$, for some $\eps>0$, then it is not clear whether $M_t^*(w^N)\succeq 0$; actually the answer depends on several factors like the value of $\eps$, the positive terms $w_I^N> 0$, the level $t$ and so on. 

The strategy that we suggest in the following uses the Gershgorin disk theorem (see e.g. \cite{varga}), which identifies a region in the plane that contains all the eigenvalues of a square matrix. Let A be an $n\times n$ matrix. For each $i$ with $1\leq i\leq n$, define the \emph{radius}
\begin{equation}\label{eq:r}
r_i= \sum_{\substack{j=1\\ j\not = i}}^n |A_{i,j}|
\end{equation}
Let $\D_i$ be the closed disc centered at $a_{ii}$ with radius $r_i$. Such a disk is called a \emph{Gershgorin disk}.
Then each eigenvalue of $A$ is in at least one of the disks. So if all the disks are located in the nonnegative plane we are guaranteed to have a PSD matrix.

\begin{theorem}[Gershgorin Disk Theorem \cite{varga}]\label{th:gershgorin}
Let A be an $n\times n$ matrix, and let $\mu$ be any eigenvalue of $A$. Then for some $i$ with $1\leq i\leq n$,
\begin{equation*}
|\mu -A_{i,i}|\leq r_i
\end{equation*}
where $r_i$ is given by \eqref{eq:r}.
\end{theorem}

\paragraph{Congruent Transformation of Gershgorin Disks.}\label{gershgorintech}

If we apply directly Theorem \ref{th:gershgorin} to matrix $M_t^*(w^N)$ this might give a loose condition: for example this happens if there is a disk with a very large radius and center close to zero. 
The strategy, described in this section, is to apply a congruent transformation $T$ with the aim to obtain a matrix with tighter disks. 
The matrices of rank one play a fundamental role.
We introduce the idea by using the following example.

\begin{example}\label{ex:strategy}
From Example~\ref{example:almostdiag}, the almost diagonal decomposition of $M_1(z)$ is:
{\small
$$
M_1^*(z^N)=\underbrace{\left(
\begin{array}{ccc}
z_{\emptyset}^N & 0 & 0 \\
0 & z_{1}^N & 0 \\
0 & 0 & z_{2}^N
\end{array}
\right)}_{D}+ z_{12}^N\underbrace{
\left(
\begin{array}{cccc}
1 & -1 & -1  \\
-1 & 1 & 1  \\
-1 & 1 & 1  
\end{array}
\right)}_{R(\{1,2\})}
$$
}
Let us assume, that according to a given solution $y$, we have $z_{1}^N<0$, whereas the other entries are positive. The question is to understand under which conditions we have $M_1(z)\succeq 0$. 

According to the assumptions, the diagonal matrix $D$ above is negative semidefinite, whereas $z_{12}^N R(\{1,2\})$ is positive semidefinite. 
A straightforward application of Gershgorin Theorem can be useless.
For example if  $z_{\emptyset}^N=z_{2}^N=1, z_{12}^N=2, z_{1}^N=-1/3$, then the Gershgorin disks of matrix $M_1^*(z^N)$ are located as follows: disks $\D_{\emptyset}$ and $\D_{\{2\}}$ are centered in $3$ and have radius $4$, whereas disk $\D_{\{1\}}$ is centered in $5/3$ and has radius $4$, i.e. the disks are not entirely located in the nonnegative plane (see the left-hand picture of Figure \ref{fig:transform}). In this situation Gershgorin Theorem gives a loose condition because the disks have too large radii.
We provide a congruent transformation of $M_1^*(z^N)$ that transforms ``useless'' Gershgorin disks to ``meaningful'' ones.


Consider the following simple congruent transformation $T$, obtained by pivoting on entry $(\{1\},\{1\})$: add the second row to the first (and symmetrically for columns) and subtract the second row from the last (and symmetrically for columns). This transforms matrix $R(\{1,2\})$ into a matrix that has zero everywhere but the pivot entry $(\{1\},\{1\})$ (this is possible because $R(\{1,2\})$ has rank one). Then we obtain the following congruent matrix.
{\small
$$
T\cdot M_1^*(z^N)\cdot T^{\top}=\left(
\begin{array}{ccc}
z_{\emptyset}^N - z_{1}^N & z_{1}^N & -z_{1}^N \\
z_{1}^N & z_{1}^N & -z_{1}^N \\
-z_{1}^N & -z_{1}^N & z_{2}^N- z_{1}^N
\end{array}
\right)+ z_{12}^N
\left(
\begin{array}{cccc}
0 & 0 & 0  \\
0 & 1 & 0  \\
0 & 0 & 0  
\end{array}
\right)
=
\left(
\begin{array}{ccc}
z_{\emptyset}^N - z_{1}^N & z_{1}^N & -z_{1}^N \\
z_{1}^N & z_{12}^N+z_{1}^N & -z_{1}^N \\
-z_{1}^N & -z_{1}^N & z_{2}^N-z_{1}^N
\end{array}
\right)
$$
}

Note that the effect of the described transformation $T$ on matrix $D$ is to \emph{perturb} the radius/center of the disks by a factor of $z_{1}^N$ (that is the value of the pivot entry $(\{1\},\{1\})$). Moreover, if we add the transformed  $R(\{1,2\})$ to the transformed $D$, this has the effect of \emph{shifting} by $z_{12}^N$ the center of disk $\D_{\{1\}}$ (i.e. the disk with negative center). 
From the final matrix, we see that if $z_{\emptyset}^N\geq |z_{1}^N|,z_{2}^N\geq |z_{1}^N|, z_{12}^N\geq 3|z_{1}^N|$ then the Gershgorin disks of the transformed matrix are in the nonnegative plane, which implies that $M_1(z)\succeq 0$. Figure \ref{fig:transform} (right-hand picture) shows the effect of the congruent transformation on the Gershgorin disks for our numerical example above.

\begin{figure}[hbtp]
    \centering
    \includegraphics[width=6cm]{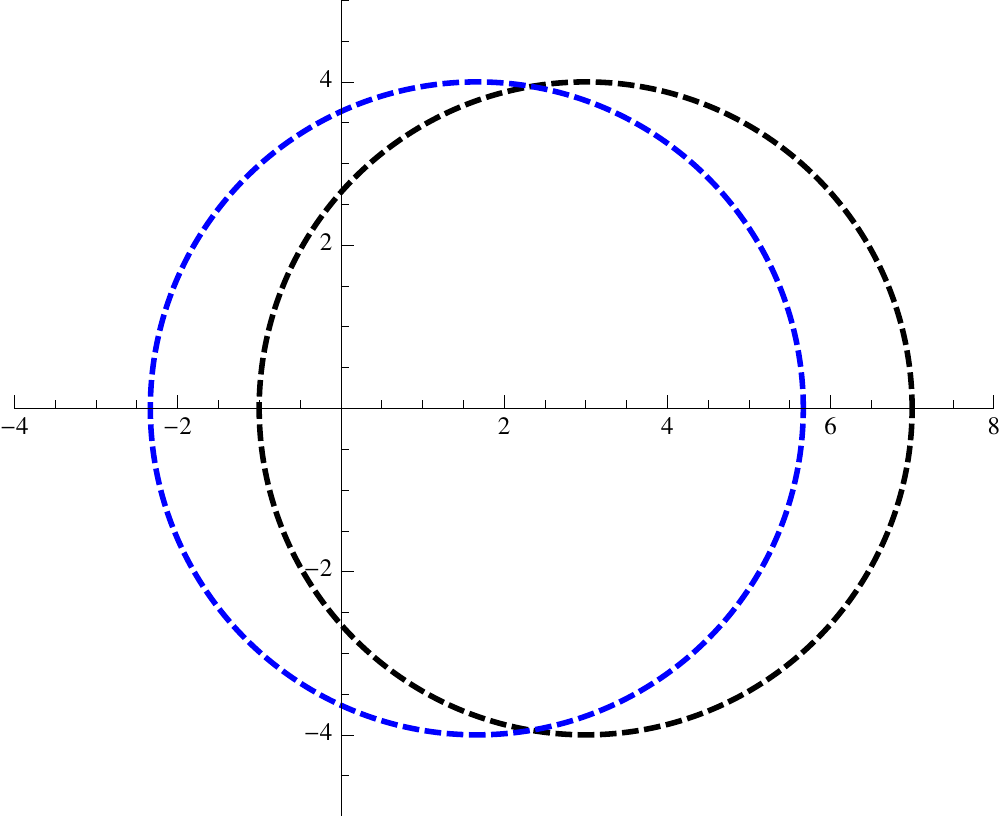}
    \includegraphics[width=6cm]{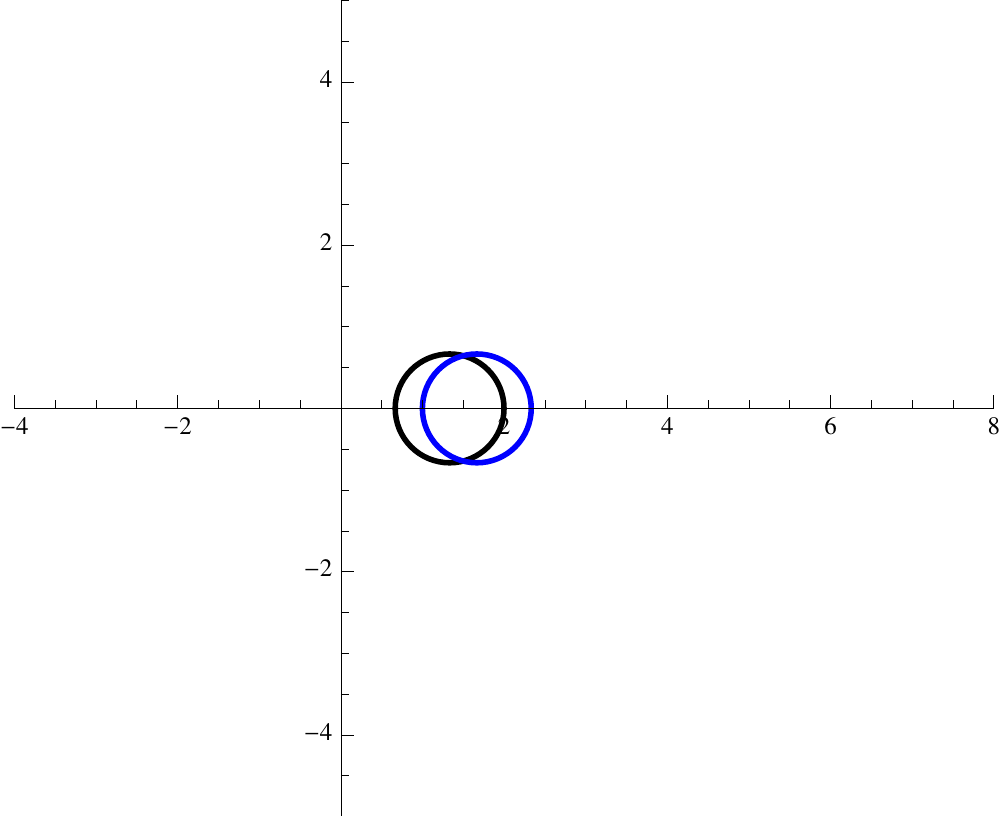}
    \caption{Gershgorin disks before and after the congruent transformation.}
  \label{fig:transform}
\end{figure}

\end{example}

Generalizing the above example, a proof of $M_t^*(w^N)\succeq 0$ can be obtained by selecting a congruent transformation $T$ and a subset $\mathcal{X}$ from the positive semidefinite matrices $\PD$ so that the Gershgorin disks of
$\M=T\cdot \left(D+\sum_{A\in \mathcal{X}} A+\sum_{B\in \ND} B\right) \cdot T^{\top}$ are located in the nonnegative plane. Note that the latter implies that $M_t^*(w^N)\succeq 0$ since it is congruent to a sum of PSD matrices. 

The congruent transformation $T$ is the concatenation of the basic congruent transformations introduced in the example (obtained by pivoting on some entries with the aim to shift/perturb disks with negative center).
In more formal terms, a \emph{basic congruent transformation} $T_S(H)$ consists of selecting a matrix from $\PD$ (or from $\ND$), say $w_H^N R(H)$, for $H\in \PS(N)\setminus \PS_t(N)$. 
Select a pivot entry $(S,S)$ that is different from zero.\footnote{Note that at the beginning, since $R(H)=G(H)G(H)^{\top}$, every entry $(S,S)$ with $S\subseteq H$ (and $|S|\leq t$) is different from zero; later, if $\{T_P(L):L\in \mathcal{T}\}$ is the set of basic transformations applied so far, entry $(S,S)$ with $S\subseteq H$ is different from zero if it is in none of the basic transformations applied so far, i.e. $S\not \subseteq L$ for every $L\in \mathcal{T}$.} Pivot on that element to obtain a congruent matrix that has zero everywhere but entry $(S,S)$ ($R(H)$ has rank one).
We call this congruent matrix the \emph{$I$-reduced form of $R(H)$} and denote it by $R_S(H)= T_S(H)\cdot R(H) \cdot T_S(H)^{\top}$. The effect of adding $R_S(H)$ to any matrix $A$ is to shift the center of disk $\D_S$ of matrix $A$ by a certain factor of $w_H^N$ (the shift is positive if $w_H^N>0$ and negative if $w_H^N<0$). The effect of transformation $T_S(H)$ on $D$ is to change the radius/center of disks $\{\D_I: I\subseteq H\}$ by a certain factor of $w_S^N$.
We recall (and remark) that every basic congruent transformation does not change the positive (negative) semi-definiteness of any matrix, so the matrices in $\PD$ (or $\ND$) remain in the same set after the transformations.

We give an application of the Gershgorin disks transformation technique to get  Lasserre integrality gaps in Section~\ref{sect:minknapsack} and~\ref{sect:knaplifted} (in the latter we use the general strategy sketched above, see also Appendix~\ref{sect:exmkp}).

\subsubsection{The Lasserre Hierarchy as a Semi-Infinite Linear Program}\label{sect:sip}
In the following we show an equivalent formulation of the Lasserre hierarchy as a semi-infinite linear program by using the almost diagonal form given by Lemma~\ref{th:lassalmostdiag}. This provides a different, non-matricial point of view that can be convenient for certain problems. We give an application of this characterization in Section~\ref{sect:tardyjobs}.

In optimization theory, \emph{semi-infinite programming} (SIP) is an optimization problem with a finite number of variables and an infinite number of constraints, or an infinite number of variables and a finite number of constraints (see e.g. \cite{GobernaL02}). It is well-known (and easy to see) that any SDP program can be written as a semi-infinite linear program. By Lemma~\ref{th:lassalmostdiag} we immediately obtain the following.

\begin{corollary}\label{th:constraintcond}
For any $w\in\mathbb{R}^{\PS(N)}$ and $t\geq 0$, we have
$M_{t}(w)\succeq 0$ if and only if for every unit vector $v\in \mathbb{R}^{\PS_t(N)}$ the following holds.
\begin{equation}\label{mkpcond}
\sum_{I:|I|\leq t} w^N_I v^2_I + \sum_{J\subseteq N: |J|\geq t+1} w^N_J  \left( \sum_{i=0}^t (-1)^{t-i}{|J|-i-1 \choose t-i} \left(\sum_{I\subset J, |I|=i} v_I\right) \right)^2  \geq 0
\end{equation}
\end{corollary}
\begin{proof}
By Lemma~\ref{th:lassalmostdiag}, we can replace any condition $M_{t}(w)\succeq 0$ with 
$M_{t}^*(w^N)\succeq 0$. Then the claim follows by the definition of PSD matrices. Indeed, let $v\in \mathbb{R}^{\PS_t(N)}$ be any eigenvector of matrix $M_{t}^*(w)$, i.e. $M_{t}^*(w) v=\lambda v$. W.l.o.g., we can assume that $v$ is a unit vector. If solution $y$ ensures $v^{\top} M_{t}^*(w) v\geq 0$ for every unit vector $v$, then we have $\lambda v^{\top} v\geq 0$, i.e. any eigenvalue $\lambda$ is nonnegative and therefore $M_{t}^*(w)\succeq 0$. It is easy to check that $v^{\top} M_{t}^*(w) v\geq 0$ is \eqref{mkpcond}.
\end{proof}



\section{Lasserre Integrality Gap for the \textsc{Min-Knapsack}}\label{sect:minknapsack}

%
 In this section we analyze the Lasserre hierarchy integrality gap for  the \textsc{Min-Knapsack} problem. The analysis uses the moment matrix in almost diagonal form and Gershgoring disk congruent transformation described in Section \ref{Sect:almostdiag}.

In the \textsc{Min-Knapsack} problem we are given a set $V$ of items with nonnegative costs $c_1,c_2\ldots$, profits $p_1,p_2,\ldots$ and demand $b$.
The goal is to select a minimum cost set of items with total profit at least the demand.
The ``standard'' linear program (LP) relaxation for the {\textsc{Min-Knapsack}} problem has this form 
  $\min\{\sum_{j\in V} c_j x_j:\sum_{j \in V} p_j x_j \geq b, x_j \in [0,1]\text{ } j\in V\}$.
%
%
%
The integrality gap of (LP) is unbounded, as the following simple instance (also used later in our results for Lasserre gap) with $n+1$ items shows:
\begin{eqnarray}\label{LP:minKnapGap}
\begin{array}{rll}
(GapLP) 
 \min \{ \sum_{i=1}^n x_i: \sum_{i=1}^{n+1} x_i \geq 1+1/P, x_i \in [0,1] \mbox{ for } i\in[n+1]\}
\end{array}
\end{eqnarray}
The optimal \emph{integral} value of ($GapLP$) is one, whereas the optimal \emph{fractional} value is $1/P$, with integrality gap $P$.

\paragraph{Our Results.} We prove the following dichotomy-type result. If we allow a ``large'' $P$ (exponential in the number of the Lasserre level), then the Lasserre hierarchy is of no help to improve the unbounded integrality  gap of ($GapLP$), even at level $(n-1)$. This analysis is tight since $\LAS_{n}(GapLP)$ admits an optimal integral solution with $n+1$ variables\footnote{The projection of $y\in \LAS_{n}(GapLP)$ on the original variables can be expressed as a convex combination of integral solutions on the first $n$ variables (see e.g.~\cite{KarlinMN11}). By selecting the solution with the lowest value and setting $x_{n+1}=1$ we obtain a feasible integral solution of value not larger than $\LAS_{n}(GapLP)$.}.
We also show that the requirement that $P$ is exponential in $n$ is necessary for having a ``large'' gap at level $(n-1)$.
These results follow easily from the definition of the Lasserre hierarchy in almost diagonal form (see Lemma~\ref{th:lassalmostdiag}) and the Gershgorin disks transformation technique (Section \ref{gershgorintech}).

\begin{theorem}\label{th:knapgapbounds}
(Integrality Gap Bounds for \textsc{Min-Knapsack})
\begin{enumerate}[(a)]
\item If $P=k\cdot 2^{2n+1}$, for any $k\geq 1$, then the integrality gap of $\LAS_{n-1}(GapLP)$ is at least $k$. \label{th:knapgapbounds(b)}
\item For any $\eps\in (0,1)$, if $P\leq (2^{n}-2)\eps$ then the integrality gap of $\LAS_{n-1}(GapLP)$ is smaller than $\frac{1}{1-\eps}$. \label{th:knapgapbounds(a)}
\end{enumerate}
\end{theorem}
(Note that the above results trivially imply that if $P=k\cdot 2^{2t+1}$, for any $t\geq 1$, then the integrality gap of $\LAS_{t-1}(GapLP)$ is at least $k$. In a different working paper, with a superset of authors, we can prove that for \emph{any} $P\geq 1$, the integrality gap of $\LAS_{t}(GapLP)$ is still $P(1-\eps)$ even when $t=O(\log^{1-\eps}n)$ for any $\eps>0$.)
%
\subsection{Proof of Theorem~\ref{th:knapgapbounds}}

Consider the following reduced instance obtained by removing variable $x_{n+1}$ and subtracting $1$ from the right-hand-side of the covering constraint in~\eqref{LP:minKnapGap}:
\begin{eqnarray}\label{LP:minKnapGap'}
\begin{array}{rll}
(GapLP') 
 \min \{ \sum_{i=1}^n x_i: \sum_{i=1}^{n} x_i - 1/P \geq 0, x_i \in [0,1] \mbox{ for } i\in[n]\}
\end{array}
\end{eqnarray}
We will use $g(x)$ to denote $\sum_{i=1}^n x_i - 1/P$ and $N=\{1, \ldots ,n\}$. Moreover, let $z=g*y$, i.e. $z_I=(g*y)_I = \sum_{i=1}^n y_{\{i\}\cup I} - y_I/P$. By the following lemma any integrality gap for $\LAS_{t}(GapLP')$ implies the same gap for $\LAS_{t}(GapLP)$. Therefore, we will focus on the reduced instance ($GapLP'$) in the following.
\begin{lemma}\label{th:gaprelation}
For any $t\in \mathbb{N}_0$, if $y'\in \LAS_{t}(GapLP')$ then $y\in \LAS_{t}(GapLP)$, where $y_I=y_{I\setminus\{n+1\}}'$ for any $I\in \PS_{2t+2}([n+1])$.
\end{lemma}
\begin{proof}[Proof Sketch.]
Let $h(x) = \sum_{i=1}^{n+1} x_i -1-1/P$. The proof follows by observing that any principal submatrix of $M_{t+1}(y)$ (or  $M_{t}(h*y)$) has either determinant equal to zero or it is a principal submatrix in $M_{t+1}(y')$ (or  $M_{t}(g*y')$).
\end{proof}

\begin{remark}
In the following we prove an unbounded integrality gap for $\LAS_{n-1}(GapLP')$.
With this aim we need to find a solution $y$ that satisfies $M_n(y)\succeq 0$ and $M_{n-1}(z)\succeq 0$. By \eqref{vars}, $M_n(y)\succeq 0$ holds if and only if $\{y^N_{I}:I\subseteq N\}$ is a probability distribution. 

So, the \emph{only} interesting part is to find a probability distribution that satisfies $M_{n-1}(z)\succeq 0$ (the constraint moment matrix), since \emph{any} probability distribution satisfies $M_n(y)\succeq 0$.
\end{remark}

\paragraph{Proof Structure.}
By Lemma~\ref{th:lassalmostdiag}, $M_{n-1}(z)$ is congruent to a matrix $M_{n-1}^*(z^N)$ that differs from the diagonal matrix $Diag(z^N,n-1)$ ($D$ for brevity) by a matrix of rank one ($G$). 
As in Example~\ref{ex:strategy}, if an entry of $D$ is negative then, by pivoting on that element, we can transform matrix $G$ to have all zeros but the pivot; This is obtained at the cost of spreading the negative entry of the diagonal matrix everywhere, and increasing therefore the disks radii of the transformed diagonal matrix by some factor of the negative entry. Now, for our instance $GapLP'$,  the first entry of $D$ is negative but it can be made arbitrarily small for ``large'' $P$. Therefore, in the resulting transformed matrix, the disks radii can be made arbitrarily small for large $P$ and the transformed $G$ shifts the negative entry by a positive value. A feasible solution, with an unbounded integrality gap, is obtained by locating Gershgorin disks in the nonnegative plane. Viceversa, if $P$ is ``small'', then the non negativity of the trace implies a ``small'' integrality gap.
Before giving the general proof, we show this for the smallest meaningful instance of $GapLP'$ with two items, by customizing the idea of Example~\ref{ex:strategy} for $GapLP'$.
\begin{example}\label{exampleknap}
In Example~\ref{example:almostdiag} we showed that $M_1(z)$ is congruent to the following almost diagonal matrix $M_1^*(z^N)$.
{\small
$$
M_1(z)\cong M_1^*(z^N)=
\left(
\begin{array}{ccc}
z_{\emptyset}-z_1-z_2+z_{12} & 0 & 0 \\
0 & z_{1}-z_{12} & 0 \\
0 & 0 & z_{2}-z_{12}
\end{array}
\right)+ z_{12}\underbrace{
\left(
\begin{array}{cccc}
1 & -1 & -1  \\
-1 & 1 & 1  \\
-1 & 1 & 1  
\end{array}
\right)}_{\mbox{matrix of rank one}}
$$
}
Consider the instance of $(GapLP')$ with two items $N=\{1,2\}$. According to this instance $M_1^*(z^N)$ is equal to:
{\small
$$
M_1^*(z^N)=
\left(
\begin{array}{ccc}
-\frac{y_{\emptyset}^{N}}{P}  & 0 & 0 \\
0 & \left(1-\frac{1}{P}\right)y_{\{1\}}^{N} & 0 \\
0 & 0 & \left(1-\frac{1}{P}\right)y_{\{2\}}^{N}
\end{array}
\right)+ \left(2-\frac{1}{P}\right)y_{\{1,2\}}^{N}
\left(
\begin{array}{cccc}
1 & -1 & -1  \\
-1 & 1 & 1  \\
-1 & 1 & 1  
\end{array}
\right)
$$
}

We are considering Lasserre at the 1-st level with 2 variables, it follows that the variable moment matrix $M_2(y)$ is a full moment matrix. Therefore, $M_2(y)\succeq 0$ is \emph{equivalent} and satisfied by assuming that $y_{\emptyset}^{N},y_{\{1\}}^{N},y_{\{2\}}^{N},y_{\{1,2\}}^{N}$ are the probabilities of the different integral solutions (see Example~\ref{examplediag}), i.e., $M_2(y)\succeq 0$ if and only if: $y_{\emptyset}^{N},y_{\{1\}}^{N},y_{\{2\}}^{N},y_{\{1,2\}}^{N}\geq 0$ and $y_{\emptyset}^{N}+y_{\{1\}}^{N}+y_{\{2\}}^{N}+y_{\{1,2\}}^{N}=1$.

Note that the Gershgorin disks of matrix $M_1^*(z^N)$ are not entirely located in the nonnegative plane. Indeed, disk $\D_{\emptyset}$ is centered in $\left(\left(2-\frac{1}{P}\right)y_{\{1,2\}}^{N}-\frac{y_{\emptyset}^{N}}{P}\right)$ and has radius $2\left(2-\frac{1}{P}\right)y_{\{1,2\}}^{N}$.

Let $\eps = \frac{y_{\emptyset}^{N}}{P}$. By pivoting on the negative entry (i.e. add the first row (column) to the second and the third rows (columns)) we obtain the following congruent matrix:
{\small
\begin{eqnarray*}
M_1^*(z^N)\cong \mathcal{M}&=&
\left(
\begin{array}{ccc}
-\eps  & -\eps & -\eps \\
-\eps & \left(1-\frac{1}{P}\right)y_{\{1\}}^{N}-\eps & -\eps \\
-\eps & -\eps & \left(1-\frac{1}{P}\right)y_{\{2\}}^{N}-\eps
\end{array}
\right)+ \left(2-\frac{1}{P}\right)y_{\{1,2\}}^{N}
\left(
\begin{array}{cccc}
1 & 0 & 0  \\
0 & 0 & 0  \\
0 & 0 & 0  
\end{array}
\right) \\
&=& \left(
\begin{array}{ccc}
\left(2-\frac{1}{P}\right)y_{\{1,2\}}^{N}-\eps  & -\eps & -\eps \\
-\eps & \left(1-\frac{1}{P}\right)y_{\{1\}}^{N}-\eps & -\eps \\
-\eps & -\eps & \left(1-\frac{1}{P}\right)y_{\{2\}}^{N}-\eps
\end{array}
\right)
\end{eqnarray*}
}
Since $0\leq y_{\emptyset}^{N}\leq 1$, by choosing $P$ ``sufficiently large'', we can make $\eps = \frac{y_{\emptyset}^{N}}{P}$ arbitrarily close to zero.
Therefore, if $P$ is ``large'' then every disk radius of the transformed matrix $\mathcal{M}$ can be made arbitrarily small. By Gershgorin's Theorem, if the diagonal entries are at least the radii, we have a feasible solution. So a feasible solution is obtained by choosing the probabilities as follows (which locate the disks of $\mathcal{M}$ in the nonnegative plane).
\begin{eqnarray*}
y_{\{1,2\}}^{N} &=& 3\eps/(2-1/P) \\
y_{\{1\}}^{N} &=& 3\eps/(1-1/P) \\
y_{\{2\}}^{N} &=& 3\eps/(1-1/P)\\
y_{\emptyset}^{N} &=& 1- y_{\{1,2\}}^{N}-y_{\{1\}}^{N}-y_{\{2\}}^{N}
\end{eqnarray*}

Moreover, note that the cost of this solution is the the expected value, i.e. the sum of the probability of each integral solution multiplied by the corresponding solution cost: 
$$y_1+y_2=\frac{3\eps}{(2-1/P)}\cdot 2+ \frac{3\eps}{(1-1/P)}\cdot 1+\frac{3\eps}{(1-1/P)}\cdot 1=O\left(\frac{1}{P}\right) $$
The optimal integral solution has value $1$. It follows that the integrality gap is unbounded by increasing $P$.

For the lower bound on $P$, consider $\mathcal{M}$ written as follows:

{$$
\mathcal{M}=
\left(
\begin{array}{ccc}
z_{\{1,2\}}+ z_{\emptyset}^N & z_{\emptyset}^N & z_{\emptyset}^N \\
z_{\emptyset}^N & z_{\{1\}}^N+ z_{\emptyset}^N &  z_{\emptyset}^N\\
z_{\emptyset}^N &  z_{\emptyset}^N & z_{\{2\}}^N+ z_{\emptyset}^N
\end{array}
\right)
$$
}
$\mathcal{M}\succeq 0$ implies that $\Tr(\mathcal{M})= 3z_{\emptyset}^N + z_{\{1\}}^N+  z_{\{2\}}^N+ z_{\{1,2\}} \geq 0$. This simplifies to $\Tr(\mathcal{M})= 2z_{\emptyset}^N + z_{\emptyset}= - \frac{2}{P} y_{\emptyset}^N +z_{\emptyset} \geq 0$ by Lemma~\ref{th:sum1}. The latter implies $ y_{\emptyset}^N\leq Pz_{\emptyset}/2$ (that generalizes to $y_{\emptyset}^N\leq \frac{Pz_{\emptyset}}{2^{n}-2}$ for $n$ items). It follows that if $P$ is ``small'' then $y_{\emptyset}^N$ is ``small''. But $y_{\emptyset}^N$ is the probability of the zero solution (i.e. the solution with all variables set to zero). Moreover, in our case, the projection of $y$ on the original variables can be expressed as a convex combination of the (infeasible) zero solution with the (feasible) positive integral solutions (i.e. the solutions with one or more variables set to $1$). A ``small'' $y_{\emptyset}^N$ implies a ``small'' integrality gap. More general and formal arguments will be provided in the following. 
\end{example}

\paragraph{(Almost) Diagonalization.}
By Lemma~\ref{th:lassalmostdiag}, $M_{n-1}(z)$ is congruent to the following matrix:

\begin{eqnarray}
M_{n-1}(z)&\cong&D + z_N \cdot G
\end{eqnarray}
where $D=Diag(z^N,n-1)$, $G=G(N)G(N)^{\top}$ and the generic entry $(I,J)$ of matrix $G$ is equal to
$G_{I,J}=(-1)^{|I|+|J|}$.

\paragraph{Pivoting.}
Let $C\in \mathbb{R}^{\PS_{n-1}(N)\times \PS_{n-1}(N)}$ be a square matrix defined as follows.
\begin{equation*}
C_{I,J}=\left\{
\begin{array}{ll}
1 & \text{if } I=J \\
(-1)^{|I|-1} & \text{if } J=\emptyset, I\not =\emptyset \\
0 & \text{otherwise}
\end{array}
\right.
\end{equation*}
Matrix $C$ is invertible (see Section~\ref{linear algebra} and Lemma~\ref{th:matrix_transf}) and it is a congruent transformation that maps $G$ to its $\emptyset$-reduced form\footnote{This congruent transformation is equivalent to pivoting on the first entry by adding to the row indexed by set $H\subseteq \PS_{n-1}(N)$ the first row (the one indexed by set $\emptyset$) multiplied by $(-1)^{|H|-1}$. Then perform the symmetric operations on the columns. This transforms $G$ into a matrix with all zeros but the first entry.} 
$D+G z_N \cong \M=C \left(D+G z_N\right) C^{\top}$.

\begin{lemma}\label{matrix}
$M_{n-1}(z) \cong \M$ where $\M\in \mathbb{R}^{\PS_{n-1}(N)\times \PS_{n-1}(N)}$ is as follows:
\begin{eqnarray*}
\M_{I,J}&=&
\left\{
\begin{array}{ll}
z_N+z^N_{\emptyset} & \text{if } I=J=\emptyset\\
z^N_{I}+z^N_{\emptyset} & \text{if } I=J\not=\emptyset\\
z^N_{\emptyset} (-1)^{|I|+1}  & \text{if } I\not=J=\emptyset\\
z^N_{\emptyset} (-1)^{|J|+1}  & \text{if } J\not=I=\emptyset\\
z^N_{\emptyset} (-1)^{|I|+|J|}  & \text{otherwise }
\end{array}
\right.
\end{eqnarray*}
\end{lemma}
\begin{proof}
$\M=C(D+G z_N)C^{\top} =C D C^{\top}+C G C^{\top} z_N$. 
%
%
\begin{itemize}
\item
$(CGC^{\top})_{I,J} = \underset{U,W\subseteq \PS_{n-1}(N)}{\sum} C_{I,U}C_{J,W}G_{U,W} = \underset{U\in \{\emptyset, I\}}{\sum} \underset{W\in \{\emptyset, J\}}{\sum} C_{I,U}C_{J,W}G_{U,W}$

If $I=J=\emptyset$ then $(CGC^{\top})_{\emptyset,\emptyset} = C_{\emptyset,\emptyset}C_{\emptyset,\emptyset}G_{\emptyset,\emptyset}=1$. Otherwise
{\begin{eqnarray*}
(CGC^{\top})_{I,J} &=& C_{I,\emptyset}C_{J,\emptyset}+ C_{I,\emptyset}C_{J,J} (-1)^{|J|} +
C_{I,I}C_{J,\emptyset} (-1)^{|I|}+ C_{I,I}C_{J,J} (-1)^{|I|+|J|}  \\
&=& (-1)^{|I|+|J|}+ (-1)^{|I|+|J|-1}+(-1)^{|I|+|J|-1}+(-1)^{|I|+|J|}=0
\end{eqnarray*}}
\item
$\left(C D C^{\top}\right)_{I,J}= \sum_{U\in \{\emptyset, I\}} C_{I,U} \sum_{W\in \{\emptyset, J\}} C_{J,W}(D)_{U,W}$.

If $I=J=\emptyset$ we have $\left(C D C^{\top}\right)_{\emptyset,\emptyset}=z_{\emptyset}^N$.
If $I=J\not =\emptyset$ we have $\left(C D C^{\top}\right)_{I,I}=\sum_{U\in \{\emptyset, I\}} C_{I,U} \sum_{W\in \{\emptyset, I\}} C_{I,W}(D)_{U,W} = z_{I}^N+z_{\emptyset}^N$. Otherwise, $I\not= J$,
we have
$\left(C D C^{\top}\right)_{I,J}= \sum_{U\in \{\emptyset, I\}} C_{I,U} \sum_{W\in \{\emptyset, J\}} C_{J,W}(D)_{U,W} 
= C_{I,\emptyset} C_{J,\emptyset} z_{\emptyset}^N $
and the claim follows by the definition of matrix $C$.
\end{itemize}
\end{proof}

\subsection{Proof of Theorem~\ref{th:knapgapbounds}\eqref{th:knapgapbounds(b)}}
In the following we prove an unbounded integrality gap for $\LAS_{n-1}(GapLP')$.
With this aim we need to find a solution $y$ that satisfies $M_n(y)\succeq 0$ and $M_{n-1}(g*y)\succeq 0$.
\begin{lemma}\label{th:feasiblesdpsolnew}
By choosing $P=k\cdot 2^{2n+1}$, for any $k\geq 1$, the following solution
\begin{eqnarray}
y^N_{I}&=& \frac{2^{n}}{P|I|-1} \qquad \forall I \subseteq N \text{ and } I\not = \emptyset \label{knapcond2}\\
y_{\emptyset}^N &=& 1-\sum_{I \subseteq N, I\not = \emptyset} y^N_{I} \label{knapcond1}
\end{eqnarray}
 guarantee $M_n(y)\succeq 0$ and $M_{n-1}(z)\succeq 0$.
\end{lemma}
\begin{proof}
By \eqref{vars}, $M_n(y)\succeq 0$ holds if and only if $y^N_{I}\geq 0$ for all $I \subseteq N$, which is guaranteed by the choice of $P$.
For $M_{n-1}(z)\succeq 0$, we use the Gershgorin's Theorem (see e.g. \cite{varga}, and Theorem~\ref{th:gershgorin}), with the congruent matrix $\M$, to obtain a set of sufficient conditions that guarantee the nonnegativity of the eigenvalues of matrix $M_{n-1}(z)$.
For each row $I$, the radius $r_I(\M)$ can be bounded as follows:
$r_I(\M) = \frac{(2^n-2)}{P} y^N_{\emptyset}  \leq \frac{(2^n-2)}{P}$, since $y^N_{\emptyset}\leq 1$ and $z^N_{\emptyset}=-\frac{1}{P} y^N_{\emptyset}$.
It follows that if $z_I^N+z_{\emptyset}^N-r_I(\M)\geq 0$, for any $\emptyset \not= I\subseteq N$, then by Theorem~\ref{th:gershgorin}, the eigenvalues of $\M$ (and therefore also of $M_{n-1}(z)$) are nonnegative and the solution feasible. So it is sufficient to have $z_I^N \geq \frac{2^n}{P}$ for any $\emptyset \not= I\subseteq N$, where $z_I^N=y_I^N (|I|-1/P)$.
\end{proof}

\paragraph{The integrality gap.}
The value of the solution given by Lemma~\ref{th:feasiblesdpsolnew} is equal to:
\begin{eqnarray*}
\sum_{i=1}^n y_i &=& \sum_{I\subseteq N} y_I^N |I| = \sum_{I\subseteq N}\frac{2^{n}}{P|I|-1} |I| \leq \frac{2^{2n+1}}{P} 
\end{eqnarray*}
By choosing $P=k\cdot 2^{2n+1}$, for any $k\geq 1$, the integrality gap is at least $k$. 

\subsection{Proof of Theorem~\ref{th:knapgapbounds}\eqref{th:knapgapbounds(a)}}
By contradiction, for some $\eps\in (0,1)$, let us assume that for $P\leq (2^n-2)\eps$ we obtain a solution $y\in\LAS_{n-1}(GapLP')$ whose value is $\sum_{i=1}^n y_i = 1-\eps$.
Consider the trace of the congruent matrix $\M$:
$\Tr(\M) = \sum_{I\in \PS_{n-1}(N)} \M_{I,I}= \sum_{I\subseteq N} z_I^N + (2^{n}-2)z_{\emptyset}^N \nonumber
= z_{\emptyset}+(2^{n}-2)(-1/P)y_{\emptyset}^N$,
where we used the equality $\sum_{I\subseteq N} z_I^N =z_{\emptyset}$ and $z_{\emptyset} = \sum_{i=1}^n y_i -1/P$ (that is smaller than 1 by the assumptions).
Since $\M\succeq 0$, it follows that $\Tr(\M)\geq 0$ and therefore
$y_{\emptyset}^N\leq \frac{Pz_{\emptyset}}{2^{n}-2}$.
%
Now, note that the objective function value can be bounded by
$\sum_{i\in N} y_i = \sum_{I\subseteq N} y_I^N|I| \geq 1-y_{\emptyset}^N$,
where we used the equality $\sum_{I\subseteq N} y_I^N =1$.
It follows that
$\sum_{i\in N} y_i \geq 1-y_{\emptyset}^N\geq 1-\frac{Pz_{\emptyset}}{2^{n}-2}> 1-\frac{P}{2^{n}-2}$.
By the assumption we have $\sum_{i\in N}y_i= 1-\eps$, and therefore
$1-\eps> 1-\frac{P}{2^{n}-2}$,
which implies
$P> (2^{n}-2)\eps$,
a contradiction.
%
%


\section{\textsc{Min-Knapsack} with Lifted Objective Function}\label{sect:knaplifted}
For \textsc{Min-Knapsack}, if we add the objective function as a constraint and impose that the value is at most one, then after one round of Lasserre the integrality gap vanishes. Indeed by adding the following constraint:
$\sum_{j=1}^n x_j \leq T$, 
and setting $T=1$ we obtain that $y_I = 0$ for any $I\subseteq N$ with $|I|>1$. The latter implies that $M_1(z)\cong Diag(z^N,1)$ which implies that any feasible fractional solution can be obtained as a convex combination of feasible integral solution, so the integrality gap is one. We can also easily show (see~\cite{KarlinMN11}) that in general the integrality gap with the lifted objective function decreases rapidly, i.e. after $1/\eps$ rounds the integrality gap is $1+O(\eps)$.

A natural question is to understand if the ``trick'' of adding the objective function can avoid the weakness of the Lasserre method when ``easy'' problems are considered. In Section~\ref{sect:tardyjobs} we show that the weakness remains even after adding the objective function. Again the almost diagonal form of the Lasserre hierarchy will play a fundamental role in the analysis. More precisely, we prove an unbounded integrality gap for a special case of the min-sum scheduling problem (see \cite{BansalP10,CheungS11}) that admits an FPTAS. The same ideas\footnote{The simple gap instances that we consider make the two problems essentially the same.} can be used for proving integrality gaps for the \textsc{Min-Multiple-Knapsack} problem (the \textsc{Min-Knapsack} variant with multiple knapsacks). 

In this section, we use the latter problem for introducing the Gershgorin disk transformation technique as explained in Section~\ref{gershgorintech} in its general form. We show this for a small instance, but the reader should have no problem to generalize it for any size, and obtain an unbounded gap for the \textsc{Min-Multiple-Knapsack} problem, as well for the min-sum scheduling problem (see Section~\ref{sect:tardyjobs}). We decided to omit this proof in full details and give an alternative proof technique that uses the Lasserre hierarchy characterization given in Section~\ref{sect:sip}.

\begin{example}\label{ex:mkp}
Consider the instance of the \textsc{Min-Multiple-Knapsack} problem with 3 knapsacks, demand $\eps=1/16$, two different items for each knapsack, with unit profit and cost. If we impose that the objective function value is not larger than two, then the linear program relaxation of the considered instance is as follows.
\begin{subequations}
\label{LP:minmultknap}
\begin{align}
  (MKP) \hspace{1cm}& \sum_{i=1}^6 x_i\leq 2 ,\label{eq:kpcardconstr}\\
  & x_1+x_2 \geq \eps \label{kpdemand1}\\
  & x_3+x_4 \geq \eps \label{kpdemand2}\\
  & x_5+x_6 \geq \eps \label{kpdemand3}\\
  & 0\leq x_{i}\leq 1, & \text{for }\ i \in[6]
  \end{align}
\end{subequations}
Note that there is no integral solution that satisfies the above constraints. In the following we show that one level of Lasserre is not sufficient for ruling out this case, giving therefore an integrality gap of $3/2$ (the integrality gap here is defined as the ratio between the optimal integral value and the objective function upper bound). By increasing the number of items and knapsacks, it is not hard to generalize this for any level $t=O(\sqrt{n})$, where $n$ is the input size, and get an unbounded integrality gap (see Section \ref{sect:tardyjobs} for a different proof technique of this claim).

With this aim, consider $\LAS_1(MKP)$ in almost diagonal form (Lemma~\ref{th:lassalmostdiag}). Let $N$ denote the set of items. Consider the 
 solution $S=\{y_I^N=\alpha : I\subseteq N, |I|\leq 2\}$ that forms a uniform probability distribution (with $y_I^N=0$ for every $|I|>2$), where $\alpha =1/(7+{6 \choose 2})=1/22$. 

Note that solution $S$ immediately satisfies $M_n(y)\succeq 0$ by \eqref{vars}. Moreover, let $g(x) = 2- \sum_{i=1}^6 x_i\geq 0$ denote the objective function constraint~\eqref{eq:kpcardconstr} and let $g(I)$ denote the value of $g(x)$ when we set to 1 all the variables in $\{x_{i}:i\in I\}$ and to zero the remaining.
Solution $S$ satisfies $M_1^*(z^N)\succeq 0$, where $z=g*y$. Indeed, it is not difficult to see that $M_1^*(z^N)$ is the sum of PSD matrices (the only positive probabilities $y_I^N$ are given to solutions with at most 2 picked items, i.e. $y_I^N g(I)\geq 0$; therefore the entries of the diagonal matrix in $M_1^*(z^N)$ are all nonnegative, and every other matrix of rank one (that is PSD) is multiplied by a nonnegative number). 

By the previous arguments, the only interesting case is to check the claim for the moment matrices of the knapsack constraints. Consider $M_1^*(z^N)\succeq 0$, where $z= g_1 *y$ and $g_1(x)= x_1 + x_2-1/P\geq 0$ (by symmetry the same hold for the other constraints). 
Let $R(J)=G(J)G(J)^{\top}$ denote the rank one matrix defined in Lemma~\ref{th:lassalmostdiag}. 
$M_1^*(z^N)$ divided by $\alpha$ is equal to:

{\small
\begin{eqnarray*}\underbrace{
\left(
\begin{array}{ccccccc}
-\eps  & 0 & 0 & 0 & 0 & 0 & 0\\
0 & 1-\eps & 0 & 0 & 0 & 0 & 0\\
0 & 0 & 1-\eps & 0 & 0 & 0 & 0 \\
0 & 0 & 0 & -\eps & 0 & 0 & 0 \\
0 & 0 & 0 & 0 & -\eps & 0 & 0 \\
0 & 0 & 0 & 0 & 0 & -\eps & 0 \\
0 & 0 & 0 & 0 & 0 & 0 & -\eps  
\end{array}
\right)}_{D}
&+& \underbrace{(2-\eps)
R(\{1,2\})
+ (1-\eps)\left(\sum_{i=3}^6 R(\{1,i\})+R(\{2,i\})\right)}_{\PD}\\
&& \underbrace{-\eps
\left(\sum_{i=3}^5 \sum_{j=i+1}^{6} R(\{i,j\})\right)}_{\ND}
\end{eqnarray*}
}

Recall that the matrices $R(\{i,j\})$ multiplied by a positive number are PSD matrices (i.e., those with $i=1$ or $j=2$), these terms belong to set $\PD$ (see notation in Section~\ref{Sect:almostdiag}). Note that all the other components of $M_1^*(z^N)$ are not PSD. 
It is easy to check that the Gershgorin disks of $M_1^*(z^N)$ are not located entirely in the nonnegative plane.

The idea is to pivot on each negative entry $(I,I)$ of the $D$-matrix to transform a certain PSD component $R(J)$ with $I\subseteq J$ into its $I$-reduced form $R^I(J)$. The transformed $D+R(J)$ can be roughly seen as the result of shifting the negative $(I,I)$-th entry of $D$ by a positive number, at the cost of spreading the negative entry value $-\eps$, and therefore increasing the disks radii in $D$ by a factor of $\eps$ (but the effect of the latter is ``small'' when $\eps$ is ``small''). On the other side this transformation does not destroy the PSD-ness of the $\PD$ components (and do not change their rank). Moreover note that the contribution of the matrices in $\ND$ is to increase the radius by some factor of $-\eps$, that is again ``small'' for sufficiently small $\eps$. After pivoting on each negative entry of $D$, we obtain a congruent matrix with positive diagonal entries and ``small'' off-diagonal entries, (plus some additional PSD matrices). For $\eps=1/16$ this congruent matrix is PSD. We provide the complete example in Appendix \ref{sect:exmkp}.

It is not very difficult to generalize this approach to prove unbounded integrality gap, for the \textsc{Min-Multiple-Knapsack} and the Min-Sum of Tardy Jobs problem (see Section~\ref{sect:tardyjobs}), at level $t=\Omega(\sqrt{n})$: solution \eqref{eq:solsched} can be shown to be feasible by iteratively pivoting on the negative entries of $D$ and obtain a final matrix with off-diagonal entries that depends only on $\eps$ (plus some additional PSD matrices). The proof of the feasibility of solution  \eqref{eq:solsched} follows by choosing $\eps$ ($1/P$ in Section~\ref{sect:tardyjobs}) ``small'' enough.  
\end{example}

\section{Lasserre Integrality Gap for the Min-Sum of Tardy Jobs}\label{sect:tardyjobs}
We consider the single machine scheduling problem to minimize the (weighted) sum of tardy jobs: we are given a set of $N$ jobs, each with a weight $w_j>0$, processing time $p_j>0$, and due date $d_j>0$. We have to sequence jobs on a single machine such that no two jobs overlap. If job $j$ completes at time $C_j$ the tardiness $T_j$ of job $j$ is $\max\{C_j-d_j,0\}$. The scheduling objective is to minimize the total weighted tardiness, i.e., $\sum_j w_j T_j$.

\paragraph{The starting LP.}
Our result is based on the following ``natural'' linear program relaxation, that is a special case of the starting LPs used in \cite{BansalP10,CheungS11} (therefore the obtained unbounded integrality gap result also holds if we apply Lasserre to the LPs used in \cite{BansalP10,CheungS11}). For each job we introduce a variable $x_j\in[0,1]$ with the intended (integral) meaning that $x_j=1$ iff job $j$ completes after its deadline, so it is a tardy job. Then, for any time $t\in\{d_1,\ldots,d_N\}$, the sum of processing times of jobs with deadlines not larger than $t$, and that completes not later than $t$, must satisfy $\sum_{j:d_j\leq t} (1-x_j)p_j \leq t$. The latter constraint can be rewritten as a capacitated covering constraint, $\sum_{j:d_j\leq t} x_jp_j \geq D_t$, where $D_t:=\sum_{j:d_j\leq t} p_j -t$ represents the \emph{demand} at time $t$. The goal is to minimize $\sum_j w_j x_j$.

\paragraph{The Gap Instance.} We consider the following instance with $N=n^2$ jobs of unit costs. (By abusing notation we will use $N$ to denote both, the set and the total number of jobs.) Jobs are partitioned into $n$ blocks $N_1, N_2,\ldots, N_n$, each with $n$ jobs. For $i\in[n]$, jobs belonging to block $N_i$ have the same processing time $P^i$, and the same deadline $d_i=n\sum_{j=1}^i P^{j}-\sum_{j=1}^i P^{j-1}$.
So the demand at time $d_i$ is $D_i=\sum_{j=1}^i P^{j-1}$, for $P>0$. For any $t\geq 0$, let $T$ be the smallest value that makes $\LAS_t\left(LP(T)\right)$ feasible, where $LP(T)$ is defined as follows:

\begin{subequations}
\label{LP:tardy}
\begin{align}
  LP(T) \hspace{1cm}& \sum_{i =1}^n \sum_{j =1}^n x_{ij}\leq T,\label{eq:cardconstr}\\
  &\sum_{i =1}^k \sum_{j =1}^n x_{ij}\cdot P^i \geq D_k, &  \text{for }\ k\in[n]\label{demand}\\
  & 0\leq x_{ij}\leq 1, & \text{for }\ i,j\in[n]
  \end{align}
\end{subequations}
Note that, for any feasible \emph{integral} solution for $LP(T)$, the smallest $T$ (i.e. the optimal integral value) can be obtained by selecting one job for each block, so the smallest $T$ for integral solutions is $n$. The \emph{integrality gap} of $\LAS_t\left(LP(T)\right)$ (or $LP(T)$)  is defined as the ratio between $n$ (i.e. the optimal integral value) and the smallest $T$ that makes $\LAS_t\left(LP(T)\right)$ (or $LP(T)$) feasible.
%
%
It is easy to check that $LP(T)$ has an integrality gap $P$ for any $P\geq 1$: For $T=n/P$, a feasible fractional solution for $LP(T)$ exists by setting $x_{ij}=\frac{1}{nP}$. 

\subsection{Unbounded Integrality Gap for the Lasserre Hierarchy}
Consider any $k\geq 1$ and $n$ such that $t+1=n/k$ is a positive integer. We show that $\LAS_t(LP(t+1))$ has a feasible solution $y$ (for a suitably large $P$). So the integrality gap is at least $k$. (Note that  at the next level, namely $t+1$, $\LAS_{t+1}(LP(t+1))$ has no feasible solution for $k>1$, which gives a tight characterization of the integrality gap threshold phenomenon.)

\paragraph{Solution Structure.}
Set $y_I = 0$ for $I\subseteq N$ with $|I|>t+1$, which implies that $M_{t+1}(y)\cong Diag(y^N,t+1)$ by Lemma~\ref{th:lassalmostdiag}; the requirement $Diag(y^N,t+1)\succeq 0$ is therefore equivalent to
$y^N_I\geq 0$ for $I\in \PS_{t+1}(N)$.
%
By setting 
\begin{eqnarray}\label{eq:solsched}
y_I^N=
\left\{ \begin{array}{ll}
\alpha = 1/ |\PS_{t+1}(N)| & \forall I\in \PS_{t+1}(N)\\
0 & otherwise 
\end{array}
\right.
\end{eqnarray}
 we have $M_{t+1}(y)\succeq 0$.

\paragraph{Constraint Moment Matrix.}
For any constraint~\eqref{eq:cardconstr} or \eqref{demand}, say $g(x)\geq 0$ ($\ell=1,\ldots,n$), we need to ensure that Condition~\eqref{mkpcond} (with $w=g*y$) is satisfied. By Solution~\eqref{eq:solsched}, Condition~\eqref{mkpcond} simplifies as follows.
\begin{lemma}\label{th:schedconstraintcond}
For any constraint $g(x)\geq 0$ (from ~\eqref{eq:cardconstr} or \eqref{demand}), Solution \eqref{eq:solsched} satisfies $M_{t}(g*y)\succeq 0$ if and only if the following conditions hold.
\begin{equation}\label{mkpcond2}
\sum_{I\in \PS_t(N)} g(I) v^2_I + \sum_{J\subseteq N:|J|=t+1} \left(\sum_{\substack{I\in \PS_t(N)\\ I\subset J}} v_I (-1)^{|I|} \right)^2 g(J)  \geq 0 \quad \forall \mbox{ unit vector } v\in \mathbb{R}^{\PS_t(N)}
\end{equation}
where $g(H)$ denote the value of $g(x)$ when we set to 1 all the variables in $\{x_{i}:i\in H\}$ and to zero the remaining.
\end{lemma}

\paragraph{The Moment Matrix for Constraint~\eqref{eq:cardconstr}.}
Constraint~\eqref{eq:cardconstr} is as follows: $g(x)= n/k - \sum_i \sum_j x_{i,j}\geq 0$. Since $t+1=n/k$ then note that for any $J\in \PS_{t+1}(N)$, we have $g(J)\geq 0$ and Condition \eqref{mkpcond2} is trivially satisfied for any vector $v$. Therefore, by Lemma \ref{th:schedconstraintcond} we have $M_{t}(g*y)\succeq 0$.

\paragraph{The Moment Matrix for Covering Constraints.}

\begin{lemma}
For any covering constraint~\eqref{demand} $g_{\ell}(x)=\sum_{i =1}^{\ell} \sum_{j =1}^n x_{ij}\cdot P^i - D_{\ell} \geq 0$, with $\ell=1,\ldots,n$, Solution \eqref{eq:solsched} satisfies Condition \eqref{mkpcond2} for $P=n^{O(t^2)}$ .
\end{lemma}
\begin{proof}
Consider the $\ell$-th covering constraint $g_{\ell}(x)\geq 0$ (see~\eqref{demand}) and the corresponding semi-infinite set of linear requirements \eqref{mkpcond2}. Then consider the following partition of $\PS_{t+1}(N)$.
\begin{eqnarray*}
A&=&\{ I\in\PS_{t+1}(N): I\cap N_{\ell} \not= \emptyset\}\\
B&=&\{I\in\PS_{t+1}(N): I\cap N_{\ell}=\emptyset\}
\end{eqnarray*} 

Note that for $S\in A$ we have $g_{\ell}(S)\geq \left(P^{\ell}-\sum_{j=1}^{\ell}P^{j-1}\right)= P^{\ell}\left(1-\frac{P^{\ell}-1}{P^{\ell}(P-1)}\right)\geq P^{\ell}\left(1-\frac{1}{P-1}\right)$. 
For $S\in B$ we have 
$g_{\ell}(S)\geq -\sum_{j=1}^{\ell}P^{j-1}\geq P^{\ell}\left(-\frac{1}{P-1}\right)$. 
Since $P>0$, by scaling $g_{\ell}(x)\geq 0$ (see~\eqref{demand}) by $P^{\ell}$, we will assume, w.l.o.g., that 
\begin{eqnarray*}
g_{\ell}(S)\geq 
\left\{
\begin{array}{ll}
1-\frac{1}{P-1} & S\in A\\
-\frac{1}{P-1} & S\in B
\end{array}
\right.
\end{eqnarray*}
Note that, since $v$ is a unit vector, we have $v_I^2\leq 1$, and for any $J\subseteq N:|J|=t+1$ the coefficient of $g_{\ell}(J)$ is bounded by $\left(\sum_{\substack{I\in \PS_t(N)\\ I\subset J}} v_I (-1)^{|I|} \right)^2\leq 2^{O(t)}$. For all unit vectors $v$ let $\beta$ denote the smallest possible total sum of the negative terms in \eqref{mkpcond2} (these are those related to $g_{\ell}(I)$ for $I\in B$). Note that $\beta\geq  -\frac{|B|2^{O(t)}}{P}= -\frac{n^{O(t)}}{P}$. 

In the following, we show that, for sufficiently large $P$, Solution \eqref{eq:solsched} satisfies \eqref{mkpcond2}. We prove this by contradiction.

Assume that it exists a unit vector $v$ such that~\eqref{mkpcond2} is not satisfied with Solution \eqref{eq:solsched}. 
We start observing that under the previous assumption the following holds  
\begin{equation}\label{eq:probcase}
\forall I\in A\cap \PS_t(N):v_I^2 = \frac{n^{O(t)}}{P}
\end{equation}
(otherwise we would have an $I\in A\cap \PS_t(N)$ such that $v_I^2 g_{\ell}(I)\geq -\beta$ contradicting the assumption that \eqref{mkpcond2} is not satisfied).
In the following we show that the previous bound on $v_I^2$ can be generalized to $v_I^2=\frac{n^{O(t^2)}}{P}$ for \emph{any} $I\in \PS_t(N)$ (under the contradiction assumption). But, by choosing $P$ such that  $v_I^2<1/n^{2t}$, for $I\in \PS_t(N)$, then we have $\sum_{I\in \PS_t(N)} v_I^2<1$, which contradicts that $v$ is a unit vector.

The claim follows by showing that $\forall I\in B\cap  \PS_t(N):v_I^2 \leq n^{O(t^2)}/P$. The proof is by induction on the size of $I$ for any $I \in B\cap  \PS_t(N)$. 
%

Consider the empty set, since $\emptyset\in B\cap  \PS_t(N)$. We show that $v_{\emptyset}^2= n^{O(t)}/P$.
With this aim, consider any $J\subseteq N_{\ell}$ with $|J|=t+1$. Note that $J\in A$, $g_{\ell}(J)\geq t+1-1/(P-1)$ and its coefficient $u_J^2=\left(\sum_{\substack{I\in \PS_t(N)\\ I\subset J}} v_I (-1)^{|I|} \right)^2$ is the square of an algebraic sum of $v_{\emptyset}$ and other terms $v_I$, all with $I\in A\cap \PS_t(N)$ and therefore $v_I^2= \frac{n^{O(t)}}{P}$. Moreover, note that $u_J^2$ is smaller than $-\beta$ (otherwise \eqref{mkpcond2} is satisfied). Therefore, we have the following bound $b_0$ for $|v_{\emptyset}|$ (here, and later, we use the loose bound that $g_{\ell}(J)\geq 1/2$ for $J\subseteq N_{\ell}$, for $P\geq 3$)
\begin{equation}\label{eq:indfirst}
|v_{\emptyset}|\leq \sqrt{-2\beta} + \sum_{\emptyset\not =I\subset J} |v_I|\leq b_0= O\left(\sqrt{-\beta}+2^{O(t)} \frac{n^{O(t)}}{\sqrt{P}}\right)= \frac{n^{O(t)}}{\sqrt{P}}
\end{equation}
which implies that $v_{\emptyset}^2=n^{O(t)}/P$.

Similarly as before, consider any singleton set $\{i\}$ with $\{i\}\in B\cap  \PS_t(N)$ and
any $J\subseteq N_{\ell}$ with $|J|=t$. Note that $J\in A$, $g_{\ell}(J)\geq t-1/(P-1)$ and its coefficient $u_J^2=\left(\sum_{\substack{I\in \PS_t(N)\\ I\subset J\cup\{i\}}} v_I (-1)^{|I|} \right)^2$ is the square of an algebraic sum of $v_{\{i\}}$, $v_{\emptyset}$ and other terms $v_I$, with $I\subseteq J$ and therefore $v_I^2= \frac{n^{O(t)}}{P}$. Moreover, again note that $u_J^2$ is smaller than $-\beta$ (otherwise \eqref{mkpcond2} is satisfied). Therefore, for any singleton set $\{i\}\in B\cap  \PS_t(N)$, we have that 
\begin{equation*}
|v_{\{i\}}|\leq |v_{\emptyset}| + \sqrt{-2\beta} + \sum_{\emptyset\not =I\subset J} |v_I| \leq 2b_0
\end{equation*}

Generalizing by induction, consider any set $S\in B\cap  \PS_t(N)$ and
any $J\subseteq N_{\ell}$ with $|J|=t+1-|S|$. We claim that $|v_{|S|}|\leq b_{|S|}$ where 
\begin{equation}\label{eq:recbound}
b_{|S|}= \sum_{i=0}^{|S|-1}\left( N^{i} b_{i}\right)+ b_0
\end{equation}
The latter \eqref{eq:recbound} follows  by induction hypothesis and by observing that again $g_{\ell}(J\cup S)u_{J\cup S}\leq -\beta$ and therefore, 
\begin{equation*}
|v_{S}|\leq \sum_{i=0}^{|S|-1} \left(\sum_{\substack{I\in B\\ |I|=i}} |v_{I}|\right) + \sqrt{-2\beta} + \sum_{I\subset J} |v_I| 
\end{equation*}

From \eqref{eq:recbound}, for any $S\in B\cap  \PS_t(N)$, we have that $|v_{S}|$ is bounded by $b_{t}=(N^{t-1}+1)b_{t-1}=N^{O(t^2)}b_0=\frac{n^{O(t^2)}}{\sqrt{P}}$.

\end{proof}

\paragraph{Acknowledgments.} I'm grateful to Joseph Cheriyan and Zhihan Gao for pointing out a mistake in an early version of the paper. I thank Adam Kurpisz and Sam Lepp\"anen for carefully reading the paper and their suggestions. I'm indebted with Ola Svensson for several stimulating discussions.

{\small
\bibliographystyle{abbrv}
\bibliography{ref}
}

\pagebreak
\appendix
\appendixpage

\section{Linear Algebra: useful facts}\label{linear algebra}

\begin{definition}[PSD]
A symmetric $n\times n$ matrix $A$ is positive semidefinite (PSD or $A\succeq 0$) if and only if for every $v\in \mathbb{R}^n$ we have $v^{\top} A v\geq 0$.
\end{definition}
\longer{
\begin{example}\label{ex:vvT}
For any vector $y$ we have $y y^{\top}\succeq 0$ since $v^{\top} (y y^{\top}) v =\sum y_i y_j v_i v_j = (\sum v_i y_i)^2 \geq~0$.
\end{example}
}
%
%
\begin{definition}\label{matrixOp}
\emph{Symmetric matrix operations} on a symmetric $n\times n$ matrix $A$ are:
\begin{enumerate}
\item Multiplying both the i-th row and i-th column by $\lambda \not= 0$.
\item Swapping the i-th and j-th column; and swapping the i-th and j-th row.
\item Adding $\lambda \times$ i-th column to j-th column and adding $\lambda \times$ i-th row to j-th row.
\end{enumerate}
We say that $A\cong B$ (read $A$ is \emph{congruent} to $B$) if and only if  $B$ is obtained from $A$ by zero or more symmetric matrix operations.
\end{definition}
The following two facts are well known (see e.g. \cite{Strang}).

\begin{lemma}\label{th:simop}
Let $A$, $B$ be symmetric. If $A \cong B$, then $A$ is PSD if and only if $B$ is PSD.
\end{lemma}

\begin{definition}\label{Def:congruence}
A \emph{congruent transformation} (or \emph{congruence transformation}) is a transformation of the form $A \rightarrow P^{\top} A P$, where $A$ and $P$ are square matrices, $P$ is invertible, and $P^{\top}$ denotes the transpose of P.
 \end{definition}

\begin{lemma}\label{th:matrix_transf}
$A\cong B$ if and only if $B= Z^{\top} A Z$ for some invertible $Z$.
\end{lemma}
\begin{proof}[Proof Sketch.]
We prove that $A\succeq 0$ if and only if $B\succeq 0$.
\begin{itemize}
\item If $B\succeq 0$ then for any vector $w$ we have $w^{\top}Bw= w^{\top} Z^{\top} A Zw \geq 0$. Set $v=Zw$ and since $Z$ is invertible for any given $v$ we can define $w=Z^{-1}v$ and we have $A\succeq 0$.
\item If $A\succeq 0$ then for any vector $v$ we have $v^{\top} A v = v^{\top} (Z^{-1})^{\top} B Z^{-1} v\geq 0$ and by setting $w=Z^{-1}v$ we obtain that $B\succeq 0$.
\end{itemize}
\end{proof}

\longer{
\paragraph{Simple facts.}
\begin{enumerate}
\item Assume that $C\in \mathbb{R}^{a\times b}$ and $D\in \mathbb{R}^{b\times b}$ then  $$(CDC^{\top})_{I,J}=\sum_{U,W}C_{I,U}C_{J,W}D_{U,W}$$
\item
\begin{equation}
\sum_{H\subseteq L} (-1)^{|H|} = \left\{
\begin{array}{ll}
1 & \text{ if } L=\emptyset \\
0 & \text{ else }
\end{array}
\right.
\end{equation}
\end{enumerate}
}
\begin{definition}[Principal Submatrix]
An $m\times m$ matrix, P, is an $m\times m$ \emph{principal submatrix} of an $n\times n$ matrix, A, if P is obtained from A by removing any $n - m$ rows and the same $n - m$ columns.
\end{definition}

\begin{lemma}
A matrix $A$ is positive semidefinite if and only if all of its principal submatrices have nonnegative determinants.
\end{lemma}
\longer{
\begin{lemma}
Let $g(x)$ be a polynomial decomposed as sum of other polynomials, i.e. $g(x)=\sum_{i\in [m]} g^{(i)}(x)$. If $M_t(g^{(i)}*y)\succeq 0$ for all $i\in [m]$ then $M_t(g*y)\succeq 0$.\end{lemma}
\begin{proof}
Note that $g*y=(\sum_{i\in [m]} g_i)*y=\sum_{i\in [m]} g_i*y$ and $M_t(\sum_{i\in [m]}g^{(i)}*y)=\sum_{i\in [m]}M_t(g^{(i)}*y)$. The claim follows by observing that $M_{t}(g*y)$ is equal to the sum of PSD matrices.
\end{proof}
}

\section{Example \ref{ex:mkp} (cont.)}\label{sect:exmkp}
Matrix $M^*_1(z)$ (normalized by $\alpha$) is as follows.
 {\tiny
\begin{eqnarray*}
&&
\underbrace{\left(
\begin{array}{ccccccc}
-\eps  & 0 & 0 & 0 & 0 & 0 & 0\\
0 & 1-\eps & 0 & 0 & 0 & 0 & 0\\
0 & 0 & 1-\eps & 0 & 0 & 0 & 0 \\
0 & 0 & 0 & -\eps & 0 & 0 & 0 \\
0 & 0 & 0 & 0 & -\eps & 0 & 0 \\
0 & 0 & 0 & 0 & 0 & -\eps & 0 \\
0 & 0 & 0 & 0 & 0 & 0 & -\eps  
\end{array}
\right)}_{D}
+ (2-\eps)\underbrace{
\left(
\begin{array}{c}
-1  \\
1 \\
1 \\
0 \\
0 \\
0 \\
0   
\end{array}
\right)
\left(
\begin{array}{c}
-1  \\
1 \\
1 \\
0 \\
0 \\
0 \\
0   
\end{array}
\right)^{\top}
}_{R(\{1,2\})}\\
&+& (1-\eps)\left(\underbrace{
\left(
\begin{array}{c}
-1  \\
1 \\
0 \\
1 \\
0 \\
0 \\
0   
\end{array}
\right)
\left(
\begin{array}{c}
-1  \\
1 \\
0 \\
1 \\
0 \\
0 \\
0   
\end{array}
\right)^{\top}
 }_{R(\{1,3\})} 
+ \underbrace{
\left(
\begin{array}{c}
-1  \\
1 \\
0 \\
0 \\
1 \\
0 \\
0   
\end{array}
\right)
\left(
\begin{array}{c}
-1  \\
1 \\
0 \\
0 \\
1 \\
0 \\
0   
\end{array}
\right)^{\top}
 }_{R(\{1,4\})} 
+  \underbrace{
\left(
\begin{array}{c}
-1  \\
1 \\
0 \\
0 \\
0 \\
1 \\
0   
\end{array}
\right)
\left(
\begin{array}{c}
-1  \\
1 \\
0 \\
0 \\
0 \\
1 \\
0   
\end{array}
\right)^{\top}
 }_{R(\{1,5\})}
+  \underbrace{
\left(
\begin{array}{c}
-1  \\
1 \\
0 \\
0 \\
0 \\
0 \\
1   
\end{array}
\right)\left(
\begin{array}{c}
-1  \\
1 \\
0 \\
0 \\
0 \\
0 \\
1   
\end{array}
\right)^{\top} }_{R(\{1,6\})} \right)\\
&+& (1-\eps)\left( \underbrace{
\left(
\begin{array}{c}
-1  \\
0 \\
1 \\
1 \\
0 \\
0 \\
0   
\end{array}
\right)
\left(
\begin{array}{c}
-1  \\
0 \\
1 \\
1 \\
0 \\
0 \\
0   
\end{array}
\right)^{\top}
 }_{R(\{2,3\})}
+  \underbrace{
\left(
\begin{array}{c}
-1  \\
0 \\
1 \\
0 \\
1 \\
0 \\
0   
\end{array}
\right)
\left(
\begin{array}{c}
-1  \\
0 \\
1 \\
0 \\
1 \\
0 \\
0   
\end{array}
\right)^{\top}
 }_{R(\{2,4\})}
+  \underbrace{
\left(
\begin{array}{c}
-1  \\
0 \\
1 \\
0 \\
0 \\
1 \\
0   
\end{array}
\right)
\left(
\begin{array}{c}
-1  \\
0 \\
1 \\
0 \\
0 \\
1 \\
0   
\end{array}
\right)^{\top}
 }_{R(\{2,5\})}
+  \underbrace{
\left(
\begin{array}{c}
-1  \\
0 \\
1 \\
0 \\
0 \\
0 \\
1   
\end{array}
\right)\left(
\begin{array}{c}
-1  \\
0 \\
1 \\
0 \\
0 \\
0 \\
1   
\end{array}
\right)^{\top} }_{R(\{2,6\})} \right)\\
&-& \eps\left( \underbrace{
\left(
\begin{array}{c}
-1  \\
0 \\
0 \\
1 \\
1 \\
0 \\
0   
\end{array}
\right)
\left(
\begin{array}{c}
-1  \\
0 \\
0 \\
1 \\
1 \\
0 \\
0   
\end{array}
\right)^{\top}
 }_{R(\{3,4\})}
+\underbrace{
\left(
\begin{array}{c}
-1  \\
0 \\
0 \\
1 \\
0 \\
1 \\
0   
\end{array}
\right)
\left(
\begin{array}{c}
-1  \\
0 \\
0 \\
1 \\
0 \\
1 \\
0   
\end{array}
\right)^{\top}
 }_{R(\{3,5\})}
+\underbrace{
\left(
\begin{array}{c}
-1  \\
0 \\
0 \\
1 \\
0 \\
0 \\
1   
\end{array}
\right)\left(
\begin{array}{c}
-1  \\
0 \\
0 \\
1 \\
0 \\
0 \\
1   
\end{array}
\right)^{\top} }_{R(\{3,6\})}
+\underbrace{
\left(
\begin{array}{c}
-1  \\
0 \\
0 \\
0 \\
1 \\
1 \\
0   
\end{array}
\right)
\left(
\begin{array}{c}
-1  \\
0 \\
0 \\
0 \\
1 \\
1 \\
0   
\end{array}
\right)^{\top}
 }_{R(\{4,5\})}
+\underbrace{
\left(
\begin{array}{c}
-1  \\
0 \\
0 \\
0 \\
1 \\
0 \\
1   
\end{array}
\right)\left(
\begin{array}{c}
-1  \\
0 \\
0 \\
0 \\
1 \\
0 \\
1   
\end{array}
\right)^{\top} }_{R(\{4,6\})}
+\underbrace{
\left(
\begin{array}{c}
-1  \\
0 \\
0 \\
0 \\
0 \\
1 \\
1   
\end{array}
\right)\left(
\begin{array}{c}
-1  \\
0 \\
0 \\
0 \\
0 \\
1 \\
1   
\end{array}
\right)^{\top} }_{R(\{5,6\})}
\right)
\end{eqnarray*}
}
Let us pivot on entry $(\emptyset,\emptyset)$ (negative in $D$) to reduce $R(\{1,2\})$ and add it to the $D$-matrix (for ease of notation we will call this sum again $D$). We obtain the following congruent matrix. We see that the new $D$ is roughly the old $D$ with  entry $(\emptyset,\emptyset)$ shifted by $(2-\eps)$ and the radius of the disks $\D_I$, with $I\subseteq \{1,2\}$, increased by some factor of the negative entry $-\eps$.

 {\tiny
\begin{eqnarray*}
&&\underbrace{
\left(
\begin{array}{ccccccc}
2  & -\eps & -\eps & 0 & 0 & 0 & 0\\
-\eps & 1-2\eps & -\eps & 0 & 0 & 0 & 0\\
-\eps & -\eps & 1-2\eps & 0 & 0 & 0 & 0 \\
0 & 0 & 0 & -\eps & 0 & 0 & 0 \\
0 & 0 & 0 & 0 & -\eps & 0 & 0 \\
0 & 0 & 0 & 0 & 0 & -\eps & 0 \\
0 & 0 & 0 & 0 & 0 & 0 & -\eps  
\end{array}
\right)}_{D}
\\
&+& (1-\eps)\left(\underbrace{
\left(
\begin{array}{c}
-1  \\
0 \\
-1 \\
1 \\
0 \\
0 \\
0   
\end{array}
\right)
\left(
\begin{array}{c}
-1  \\
0 \\
-1 \\
1 \\
0 \\
0 \\
0   
\end{array}
\right)^{\top}
 }_{R(\{1,3\})} 
+ \underbrace{
\left(
\begin{array}{c}
-1  \\
0 \\
-1 \\
0 \\
1 \\
0 \\
0   
\end{array}
\right)
\left(
\begin{array}{c}
-1  \\
0 \\
-1 \\
0 \\
1 \\
0 \\
0   
\end{array}
\right)^{\top}
 }_{R(\{1,4\})} 
+  \underbrace{
\left(
\begin{array}{c}
-1  \\
0 \\
-1 \\
0 \\
0 \\
1 \\
0   
\end{array}
\right)
\left(
\begin{array}{c}
-1  \\
0 \\
-1 \\
0 \\
0 \\
1 \\
0   
\end{array}
\right)^{\top}
 }_{R(\{1,5\})}
+  \underbrace{
\left(
\begin{array}{c}
-1  \\
0 \\
-1 \\
0 \\
0 \\
0 \\
1   
\end{array}
\right)\left(
\begin{array}{c}
-1  \\
0 \\
-1 \\
0 \\
0 \\
0 \\
1   
\end{array}
\right)^{\top} }_{R(\{1,6\})} \right)\\
&+& (1-\eps)\left( \underbrace{
\left(
\begin{array}{c}
-1  \\
-1 \\
0 \\
1 \\
0 \\
0 \\
0   
\end{array}
\right)
\left(
\begin{array}{c}
-1  \\
-1 \\
0 \\
1 \\
0 \\
0 \\
0   
\end{array}
\right)^{\top}
 }_{R(\{2,3\})}
+  \underbrace{
\left(
\begin{array}{c}
-1  \\
-1 \\
0 \\
0 \\
1 \\
0 \\
0   
\end{array}
\right)
\left(
\begin{array}{c}
-1  \\
-1 \\
0 \\
0 \\
1 \\
0 \\
0   
\end{array}
\right)^{\top}
 }_{R(\{2,4\})}
+  \underbrace{
\left(
\begin{array}{c}
-1  \\
-1 \\
0 \\
0 \\
0 \\
1 \\
0   
\end{array}
\right)
\left(
\begin{array}{c}
-1  \\
-1 \\
0 \\
0 \\
0 \\
1 \\
0   
\end{array}
\right)^{\top}
 }_{R(\{2,5\})}
+  \underbrace{
\left(
\begin{array}{c}
-1  \\
-1 \\
0 \\
0 \\
0 \\
0 \\
1   
\end{array}
\right)\left(
\begin{array}{c}
-1  \\
-1 \\
0 \\
0 \\
0 \\
0 \\
1   
\end{array}
\right)^{\top} }_{R(\{2,6\})} \right)\\
&-& \eps\left( \underbrace{
\left(
\begin{array}{c}
-1  \\
-1 \\
-1 \\
1 \\
1 \\
0 \\
0   
\end{array}
\right)
\left(
\begin{array}{c}
-1  \\
-1 \\
-1 \\
1 \\
1 \\
0 \\
0   
\end{array}
\right)^{\top}
 }_{R(\{3,4\})}
+\underbrace{
\left(
\begin{array}{c}
-1  \\
-1 \\
-1 \\
1 \\
0 \\
1 \\
0   
\end{array}
\right)
\left(
\begin{array}{c}
-1  \\
-1 \\
-1 \\
1 \\
0 \\
1 \\
0   
\end{array}
\right)^{\top}
 }_{R(\{3,5\})}
+\underbrace{
\left(
\begin{array}{c}
-1  \\
-1 \\
-1 \\
1 \\
0 \\
0 \\
1   
\end{array}
\right)\left(
\begin{array}{c}
-1  \\
-1 \\
-1 \\
1 \\
0 \\
0 \\
1   
\end{array}
\right)^{\top} }_{R(\{3,6\})}
+\underbrace{
\left(
\begin{array}{c}
-1  \\
-1 \\
-1 \\
0 \\
1 \\
1 \\
0   
\end{array}
\right)
\left(
\begin{array}{c}
-1  \\
-1 \\
-1 \\
0 \\
1 \\
1 \\
0   
\end{array}
\right)^{\top}
 }_{R(\{4,5\})}
+\underbrace{
\left(
\begin{array}{c}
-1  \\
-1 \\
-1 \\
0 \\
1 \\
0 \\
1   
\end{array}
\right)\left(
\begin{array}{c}
-1  \\
-1 \\
-1 \\
0 \\
1 \\
0 \\
1   
\end{array}
\right)^{\top} }_{R(\{4,6\})}
+\underbrace{
\left(
\begin{array}{c}
-1  \\
-1 \\
-1 \\
0 \\
0 \\
1 \\
1   
\end{array}
\right)\left(
\begin{array}{c}
-1  \\
-1 \\
-1 \\
0 \\
0 \\
1 \\
1   
\end{array}
\right)^{\top} }_{R(\{5,6\})}
\right)
\end{eqnarray*}
}

Now, pivot on entry $(3,3)$ to reduce $R(\{1,3\})$, and add it to the $D$-matrix. We obtain the following congruent matrix.

 {\tiny
\begin{eqnarray*}
&&\underbrace{
\left(
\begin{array}{ccccccc}
2-\eps  & -\eps & -2\eps & -\eps & 0 & 0 & 0\\
-\eps & 1-2\eps & -\eps & 0 & 0 & 0 & 0\\
-2\eps & -\eps & 1-3\eps & -\eps & 0 & 0 & 0 \\
-\eps & 0 & -\eps & 1-2\eps & 0 & 0 & 0 \\
0 & 0 & 0 & 0 & -\eps & 0 & 0 \\
0 & 0 & 0 & 0 & 0 & -\eps & 0 \\
0 & 0 & 0 & 0 & 0 & 0 & -\eps  
\end{array}
\right)}_{D}
\\
&+& (1-\eps)\left( \underbrace{
\left(
\begin{array}{c}
-1  \\
0 \\
-1 \\
0 \\
1 \\
0 \\
0   
\end{array}
\right)
\left(
\begin{array}{c}
-1  \\
0 \\
-1 \\
0 \\
1 \\
0 \\
0   
\end{array}
\right)^{\top}
 }_{R(\{1,4\})} 
+  \underbrace{
\left(
\begin{array}{c}
-1  \\
0 \\
-1 \\
0 \\
0 \\
1 \\
0   
\end{array}
\right)
\left(
\begin{array}{c}
-1  \\
0 \\
-1 \\
0 \\
0 \\
1 \\
0   
\end{array}
\right)^{\top}
 }_{R(\{1,5\})}
+  \underbrace{
\left(
\begin{array}{c}
-1  \\
0 \\
-1 \\
0 \\
0 \\
0 \\
1   
\end{array}
\right)\left(
\begin{array}{c}
-1  \\
0 \\
-1 \\
0 \\
0 \\
0 \\
1   
\end{array}
\right)^{\top} }_{R(\{1,6\})} \right)\\
&+& (1-\eps)\left( \underbrace{
\left(
\begin{array}{c}
0  \\
-1 \\
1 \\
1 \\
0 \\
0 \\
0   
\end{array}
\right)
\left(
\begin{array}{c}
0  \\
-1 \\
1 \\
1 \\
0 \\
0 \\
0   
\end{array}
\right)^{\top}
 }_{R(\{2,3\})}
+  \underbrace{
\left(
\begin{array}{c}
-1  \\
-1 \\
0 \\
0 \\
1 \\
0 \\
0   
\end{array}
\right)
\left(
\begin{array}{c}
-1  \\
-1 \\
0 \\
0 \\
1 \\
0 \\
0   
\end{array}
\right)^{\top}
 }_{R(\{2,4\})}
+  \underbrace{
\left(
\begin{array}{c}
-1  \\
-1 \\
0 \\
0 \\
0 \\
1 \\
0   
\end{array}
\right)
\left(
\begin{array}{c}
-1  \\
-1 \\
0 \\
0 \\
0 \\
1 \\
0   
\end{array}
\right)^{\top}
 }_{R(\{2,5\})}
+  \underbrace{
\left(
\begin{array}{c}
-1  \\
-1 \\
0 \\
0 \\
0 \\
0 \\
1   
\end{array}
\right)\left(
\begin{array}{c}
-1  \\
-1 \\
0 \\
0 \\
0 \\
0 \\
1   
\end{array}
\right)^{\top} }_{R(\{2,6\})} \right)\\
&-& \eps\left( \underbrace{
\left(
\begin{array}{c}
0  \\
-1 \\
0 \\
1 \\
1 \\
0 \\
0   
\end{array}
\right)
\left(
\begin{array}{c}
0  \\
-1 \\
0 \\
1 \\
1 \\
0 \\
0   
\end{array}
\right)^{\top}
 }_{R(\{3,4\})}
+\underbrace{
\left(
\begin{array}{c}
0  \\
-1 \\
0 \\
1 \\
0 \\
1 \\
0   
\end{array}
\right)
\left(
\begin{array}{c}
0  \\
-1 \\
0 \\
1 \\
0 \\
1 \\
0   
\end{array}
\right)^{\top}
 }_{R(\{3,5\})}
+\underbrace{
\left(
\begin{array}{c}
0  \\
-1 \\
0 \\
1 \\
0 \\
0 \\
1   
\end{array}
\right)\left(
\begin{array}{c}
0  \\
-1 \\
0 \\
1 \\
0 \\
0 \\
1   
\end{array}
\right)^{\top} }_{R(\{3,6\})}
+\underbrace{
\left(
\begin{array}{c}
-1  \\
-1 \\
-1 \\
0 \\
1 \\
1 \\
0   
\end{array}
\right)
\left(
\begin{array}{c}
-1  \\
-1 \\
-1 \\
0 \\
1 \\
1 \\
0   
\end{array}
\right)^{\top}
 }_{R(\{4,5\})}
+\underbrace{
\left(
\begin{array}{c}
-1  \\
-1 \\
-1 \\
0 \\
1 \\
0 \\
1   
\end{array}
\right)\left(
\begin{array}{c}
-1  \\
-1 \\
-1 \\
0 \\
1 \\
0 \\
1   
\end{array}
\right)^{\top} }_{R(\{4,6\})}
+\underbrace{
\left(
\begin{array}{c}
-1  \\
-1 \\
-1 \\
0 \\
0 \\
1 \\
1   
\end{array}
\right)\left(
\begin{array}{c}
-1  \\
-1 \\
-1 \\
0 \\
0 \\
1 \\
1   
\end{array}
\right)^{\top} }_{R(\{5,6\})}
\right)
\end{eqnarray*}
}

Now pivot on entry $(4,4)$ to reduce $R(\{2,4\})$, and add it to the $D$-matrix. We obtain the following congruent matrix.

 {\tiny
\begin{eqnarray*}
&&\underbrace{
\left(
\begin{array}{ccccccc}
2-2\eps  & -2\eps & -2\eps & -\eps & -\eps & 0 & 0\\
-2\eps & 1-3\eps & -\eps & 0 & -\eps & 0 & 0\\
-2\eps & -\eps & 1-3\eps & -\eps & 0 & 0 & 0 \\
-\eps & 0 & -\eps & 1-2\eps & 0 & 0 & 0 \\
-\eps & -\eps & 0 & 0 & 1-2\eps & 0 & 0 \\
0 & 0 & 0 & 0 & 0 & -\eps & 0 \\
0 & 0 & 0 & 0 & 0 & 0 & -\eps  
\end{array}
\right)}_{D}
\\
&+& (1-\eps)\left( \underbrace{
\left(
\begin{array}{c}
1  \\
1 \\
-1 \\
0 \\
1 \\
0 \\
0   
\end{array}
\right)
\left(
\begin{array}{c}
1  \\
1 \\
-1 \\
0 \\
1 \\
0 \\
0   
\end{array}
\right)^{\top}
 }_{R(\{1,4\})} 
+  \underbrace{
\left(
\begin{array}{c}
-1  \\
0 \\
-1 \\
0 \\
0 \\
1 \\
0   
\end{array}
\right)
\left(
\begin{array}{c}
-1  \\
0 \\
-1 \\
0 \\
0 \\
1 \\
0   
\end{array}
\right)^{\top}
 }_{R(\{1,5\})}
+  \underbrace{
\left(
\begin{array}{c}
-1  \\
0 \\
-1 \\
0 \\
0 \\
0 \\
1   
\end{array}
\right)\left(
\begin{array}{c}
-1  \\
0 \\
-1 \\
0 \\
0 \\
0 \\
1   
\end{array}
\right)^{\top} }_{R(\{1,6\})} \right)\\
&+& (1-\eps)\left( \underbrace{
\left(
\begin{array}{c}
0  \\
-1 \\
1 \\
1 \\
0 \\
0 \\
0   
\end{array}
\right)
\left(
\begin{array}{c}
0  \\
-1 \\
1 \\
1 \\
0 \\
0 \\
0   
\end{array}
\right)^{\top}
 }_{R(\{2,3\})}
+  \underbrace{
\left(
\begin{array}{c}
-1  \\
-1 \\
0 \\
0 \\
0 \\
1 \\
0   
\end{array}
\right)
\left(
\begin{array}{c}
-1  \\
-1 \\
0 \\
0 \\
0 \\
1 \\
0   
\end{array}
\right)^{\top}
 }_{R(\{2,5\})}
+  \underbrace{
\left(
\begin{array}{c}
-1  \\
-1 \\
0 \\
0 \\
0 \\
0 \\
1   
\end{array}
\right)\left(
\begin{array}{c}
-1  \\
-1 \\
0 \\
0 \\
0 \\
0 \\
1   
\end{array}
\right)^{\top} }_{R(\{2,6\})} \right)\\
&-& \eps\left( \underbrace{
\left(
\begin{array}{c}
1  \\
0 \\
0 \\
1 \\
1 \\
0 \\
0   
\end{array}
\right)
\left(
\begin{array}{c}
1  \\
0 \\
0 \\
1 \\
1 \\
0 \\
0   
\end{array}
\right)^{\top}
 }_{R(\{3,4\})}
+\underbrace{
\left(
\begin{array}{c}
0  \\
-1 \\
0 \\
1 \\
0 \\
1 \\
0   
\end{array}
\right)
\left(
\begin{array}{c}
0  \\
-1 \\
0 \\
1 \\
0 \\
1 \\
0   
\end{array}
\right)^{\top}
 }_{R(\{3,5\})}
+\underbrace{
\left(
\begin{array}{c}
0  \\
-1 \\
0 \\
1 \\
0 \\
0 \\
1   
\end{array}
\right)\left(
\begin{array}{c}
0  \\
-1 \\
0 \\
1 \\
0 \\
0 \\
1   
\end{array}
\right)^{\top} }_{R(\{3,6\})}
+\underbrace{
\left(
\begin{array}{c}
0  \\
0 \\
-1 \\
0 \\
1 \\
1 \\
0   
\end{array}
\right)
\left(
\begin{array}{c}
0  \\
0 \\
-1 \\
0 \\
1 \\
1 \\
0   
\end{array}
\right)^{\top}
 }_{R(\{4,5\})}
+\underbrace{
\left(
\begin{array}{c}
0  \\
0 \\
-1 \\
0 \\
1 \\
0 \\
1   
\end{array}
\right)\left(
\begin{array}{c}
0  \\
0 \\
-1 \\
0 \\
1 \\
0 \\
1   
\end{array}
\right)^{\top} }_{R(\{4,6\})}
+\underbrace{
\left(
\begin{array}{c}
-1  \\
-1 \\
-1 \\
0 \\
0 \\
1 \\
1   
\end{array}
\right)\left(
\begin{array}{c}
-1  \\
-1 \\
-1 \\
0 \\
0 \\
1 \\
1   
\end{array}
\right)^{\top} }_{R(\{5,6\})}
\right)
\end{eqnarray*}
}

Now, pivot on entry $(5,5)$ to reduce $R(\{1,5\})$, and add it to the $D$-matrix. We obtain the following congruent matrix.

 {\tiny
\begin{eqnarray*}
&&\underbrace{
\left(
\begin{array}{ccccccc}
2-3\eps  & -2\eps & -3\eps & -\eps & -\eps & -\eps & 0\\
-2\eps & 1-3\eps & -\eps & 0 & -\eps & 0 & 0\\
-3\eps & -\eps & 1-4\eps & -\eps & 0 & -\eps & 0 \\
-\eps & 0 & -\eps & 1-2\eps & 0 & 0 & 0 \\
-\eps & -\eps & 0 & 0 & 1-2\eps & 0 & 0 \\
-\eps & 0 & -\eps & 0 & 0 & 1-2\eps & 0 \\
0 & 0 & 0 & 0 & 0 & 0 & -\eps  
\end{array}
\right)}_{D}
\\
&+& (1-\eps)\left( \underbrace{
\left(
\begin{array}{c}
1  \\
1 \\
-1 \\
0 \\
1 \\
0 \\
0   
\end{array}
\right)
\left(
\begin{array}{c}
1  \\
1 \\
-1 \\
0 \\
1 \\
0 \\
0   
\end{array}
\right)^{\top}
 }_{R(\{1,4\})} 
+   \underbrace{
\left(
\begin{array}{c}
-1  \\
0 \\
-1 \\
0 \\
0 \\
0 \\
1   
\end{array}
\right)\left(
\begin{array}{c}
-1  \\
0 \\
-1 \\
0 \\
0 \\
0 \\
1   
\end{array}
\right)^{\top} }_{R(\{1,6\})} \right)\\
&+& (1-\eps)\left( \underbrace{
\left(
\begin{array}{c}
0  \\
-1 \\
1 \\
1 \\
0 \\
0 \\
0   
\end{array}
\right)
\left(
\begin{array}{c}
0  \\
-1 \\
1 \\
1 \\
0 \\
0 \\
0   
\end{array}
\right)^{\top}
 }_{R(\{2,3\})}
+  \underbrace{
\left(
\begin{array}{c}
0  \\
-1 \\
1 \\
0 \\
0 \\
1 \\
0   
\end{array}
\right)
\left(
\begin{array}{c}
0  \\
-1 \\
1 \\
0 \\
0 \\
1 \\
0   
\end{array}
\right)^{\top}
 }_{R(\{2,5\})}
+  \underbrace{
\left(
\begin{array}{c}
-1  \\
-1 \\
0 \\
0 \\
0 \\
0 \\
1   
\end{array}
\right)\left(
\begin{array}{c}
-1  \\
-1 \\
0 \\
0 \\
0 \\
0 \\
1   
\end{array}
\right)^{\top} }_{R(\{2,6\})} \right)\\
&-& \eps\left( \underbrace{
\left(
\begin{array}{c}
1  \\
0 \\
0 \\
1 \\
1 \\
0 \\
0   
\end{array}
\right)
\left(
\begin{array}{c}
1  \\
0 \\
0 \\
1 \\
1 \\
0 \\
0   
\end{array}
\right)^{\top}
 }_{R(\{3,4\})}
+\underbrace{
\left(
\begin{array}{c}
1  \\
-1 \\
1 \\
1 \\
0 \\
1 \\
0   
\end{array}
\right)
\left(
\begin{array}{c}
1  \\
-1 \\
1 \\
1 \\
0 \\
1 \\
0   
\end{array}
\right)^{\top}
 }_{R(\{3,5\})}
+\underbrace{
\left(
\begin{array}{c}
0  \\
-1 \\
0 \\
1 \\
0 \\
0 \\
1   
\end{array}
\right)\left(
\begin{array}{c}
0  \\
-1 \\
0 \\
1 \\
0 \\
0 \\
1   
\end{array}
\right)^{\top} }_{R(\{3,6\})}
+\underbrace{
\left(
\begin{array}{c}
1  \\
0 \\
0 \\
0 \\
1 \\
1 \\
0   
\end{array}
\right)
\left(
\begin{array}{c}
1  \\
0 \\
0 \\
0 \\
1 \\
1 \\
0   
\end{array}
\right)^{\top}
 }_{R(\{4,5\})}
+\underbrace{
\left(
\begin{array}{c}
0  \\
0 \\
-1 \\
0 \\
1 \\
0 \\
1   
\end{array}
\right)\left(
\begin{array}{c}
0  \\
0 \\
-1 \\
0 \\
1 \\
0 \\
1   
\end{array}
\right)^{\top} }_{R(\{4,6\})}
+\underbrace{
\left(
\begin{array}{c}
0  \\
-1 \\
0 \\
0 \\
0 \\
1 \\
1   
\end{array}
\right)\left(
\begin{array}{c}
0  \\
-1 \\
0 \\
0 \\
0 \\
1 \\
1   
\end{array}
\right)^{\top} }_{R(\{5,6\})}
\right)
\end{eqnarray*}
}

Pivot on the last negative entry $(6,6)$ of the $D$-matrix to reduce $R(\{2,6\})$, and add it to the $D$-matrix. We obtain the following congruent matrix.

 {\tiny
\begin{eqnarray*}
&&\underbrace{
\left(
\begin{array}{ccccccc}
2-4\eps  & -3\eps & -3\eps & -\eps & -\eps & -\eps & -\eps\\
-3\eps & 1-4\eps & -\eps & 0 & -\eps & 0 & -\eps\\
-3\eps & -\eps & 1-4\eps & -\eps & 0 & -\eps & 0\\
-\eps & 0 & -\eps & 1-2\eps & 0 & 0 & 0 \\
-\eps & -\eps & 0 & 0 & 1-2\eps & 0 & 0 \\
-\eps & 0 & -\eps & 0 & 0 & 1-2\eps & 0 \\
-\eps & -\eps & 0 & 0 & 0 & 0 & 1-2\eps  
\end{array}
\right)}_{D}
\\
&+& (1-\eps)\left( \underbrace{
\left(
\begin{array}{c}
1  \\
1 \\
-1 \\
0 \\
1 \\
0 \\
0   
\end{array}
\right)
\left(
\begin{array}{c}
1  \\
1 \\
-1 \\
0 \\
1 \\
0 \\
0   
\end{array}
\right)^{\top}
 }_{R(\{1,4\})} 
+   \underbrace{
\left(
\begin{array}{c}
0  \\
1 \\
-1 \\
0 \\
0 \\
0 \\
1   
\end{array}
\right)\left(
\begin{array}{c}
0  \\
1 \\
-1 \\
0 \\
0 \\
0 \\
1   
\end{array}
\right)^{\top} }_{R(\{1,6\})} \right)\\
&+& (1-\eps)\left( \underbrace{
\left(
\begin{array}{c}
0  \\
-1 \\
1 \\
1 \\
0 \\
0 \\
0   
\end{array}
\right)
\left(
\begin{array}{c}
0  \\
-1 \\
1 \\
1 \\
0 \\
0 \\
0   
\end{array}
\right)^{\top}
 }_{R(\{2,3\})}
+  \underbrace{
\left(
\begin{array}{c}
0  \\
-1 \\
1 \\
0 \\
0 \\
1 \\
0   
\end{array}
\right)
\left(
\begin{array}{c}
0  \\
-1 \\
1 \\
0 \\
0 \\
1 \\
0   
\end{array}
\right)^{\top}
 }_{R(\{2,5\})} \right)\\
&-& \eps\left( \underbrace{
\left(
\begin{array}{c}
1  \\
0 \\
0 \\
1 \\
1 \\
0 \\
0   
\end{array}
\right)
\left(
\begin{array}{c}
1  \\
0 \\
0 \\
1 \\
1 \\
0 \\
0   
\end{array}
\right)^{\top}
 }_{R(\{3,4\})}
+\underbrace{
\left(
\begin{array}{c}
1  \\
-1 \\
1 \\
1 \\
0 \\
1 \\
0   
\end{array}
\right)
\left(
\begin{array}{c}
1  \\
-1 \\
1 \\
1 \\
0 \\
1 \\
0   
\end{array}
\right)^{\top}
 }_{R(\{3,5\})}
+\underbrace{
\left(
\begin{array}{c}
1  \\
0 \\
0 \\
1 \\
0 \\
0 \\
1   
\end{array}
\right)\left(
\begin{array}{c}
1  \\
0 \\
0 \\
1 \\
0 \\
0 \\
1   
\end{array}
\right)^{\top} }_{R(\{3,6\})}
+\underbrace{
\left(
\begin{array}{c}
1  \\
0 \\
0 \\
0 \\
1 \\
1 \\
0   
\end{array}
\right)
\left(
\begin{array}{c}
1  \\
0 \\
0 \\
0 \\
1 \\
1 \\
0   
\end{array}
\right)^{\top}
 }_{R(\{4,5\})}
+\underbrace{
\left(
\begin{array}{c}
1  \\
1 \\
-1 \\
0 \\
1 \\
0 \\
1   
\end{array}
\right)\left(
\begin{array}{c}
1  \\
1 \\
-1 \\
0 \\
1 \\
0 \\
1   
\end{array}
\right)^{\top} }_{R(\{4,6\})}
+\underbrace{
\left(
\begin{array}{c}
1  \\
0 \\
0 \\
0 \\
0 \\
1 \\
1   
\end{array}
\right)\left(
\begin{array}{c}
1  \\
0 \\
0 \\
0 \\
0 \\
1 \\
1   
\end{array}
\right)^{\top} }_{R(\{5,6\})}
\right)
\end{eqnarray*}
}
Now, the sum of the $\ND$ matrices is

\begin{eqnarray*}
-\eps
\left(
\begin{array}{ccccccc}
6  & 0 & 0 & 3 & 3 & 3 & 3\\
0 & 2 & -2 & -1 & 1 & -1 & 1\\
0 & -2 & 2 & 1 & -1 & 1 & -1\\
3 & -1 & 1 & 3 & 1 & 1 & 1\\ 
3 &  1 &  -1 &  1 &  3 &  1 &  1\\ 
3 &  -1 &  1 &  1 &  1 &  3 &  1\\  
3 &  1 &  -1 &  1 &  1 &  1 &  3
\end{array}
\right)
\end{eqnarray*}
and we add it to the $D$-matrix and obtain:

 {\tiny
\begin{eqnarray*}
&&\underbrace{
\left(
\begin{array}{ccccccc}
2-10\eps  & -3\eps & -3\eps & -4\eps & -4\eps & -4\eps & -4\eps\\
-3\eps & 1-6\eps & \eps & \eps & -2\eps & \eps & -2\eps\\
-3\eps & \eps & 1-6\eps & -2\eps & \eps & -2\eps & \eps\\
-4\eps & \eps & -2\eps & 1-5\eps & -\eps & -\eps & -\eps \\
-4\eps & -2\eps & \eps & -\eps & 1-5\eps & -\eps & -\eps \\
-4\eps & \eps & -2\eps & -\eps & -\eps & 1-5\eps & -\eps \\
-4\eps & -2\eps & \eps & -\eps & -\eps & -\eps & 1-5\eps  
\end{array}
\right)}_{D}
\\
&+& (1-\eps)\left( \underbrace{
\left(
\begin{array}{c}
1  \\
1 \\
-1 \\
0 \\
1 \\
0 \\
0   
\end{array}
\right)
\left(
\begin{array}{c}
1  \\
1 \\
-1 \\
0 \\
1 \\
0 \\
0   
\end{array}
\right)^{\top}
 }_{R(\{1,4\})} 
+   \underbrace{
\left(
\begin{array}{c}
0  \\
1 \\
-1 \\
0 \\
0 \\
0 \\
1   
\end{array}
\right)\left(
\begin{array}{c}
0  \\
1 \\
-1 \\
0 \\
0 \\
0 \\
1   
\end{array}
\right)^{\top} }_{R(\{1,6\})} \right)\\
&+& (1-\eps)\left( \underbrace{
\left(
\begin{array}{c}
0  \\
-1 \\
1 \\
1 \\
0 \\
0 \\
0   
\end{array}
\right)
\left(
\begin{array}{c}
0  \\
-1 \\
1 \\
1 \\
0 \\
0 \\
0   
\end{array}
\right)^{\top}
 }_{R(\{2,3\})}
+  \underbrace{
\left(
\begin{array}{c}
0  \\
-1 \\
1 \\
0 \\
0 \\
1 \\
0   
\end{array}
\right)
\left(
\begin{array}{c}
0  \\
-1 \\
1 \\
0 \\
0 \\
1 \\
0   
\end{array}
\right)^{\top}
 }_{R(\{2,5\})} \right)\\
\end{eqnarray*}
}

We see that $\eps=1/16$ locates Gershgorin disks of matrix $D$ in the nonnegative plane, whereas the other matrices are PSD. This shows that the suggested solution is feasible.

\end{document}